\newif\ifpets

\petsfalse 

\ifpets
\documentclass[USenglish,oneside,twocolumn]{article}
\else 
\documentclass[USenglish,oneside]{article}
\usepackage[letterpaper, portrait, margin=1in]{geometry}
\fi

\usepackage[utf8]{inputenc}
\ifpets
\usepackage[big]{dgruyter_NEW}

\DOI{foobar}

\cclogo{\includegraphics{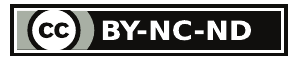}}
\fi

\usepackage{amsmath}
\usepackage{amssymb, amsthm}
\usepackage{graphicx}

\newtheorem{proposition}{Proposition}
\newtheorem{lemma}{Lemma}
\newtheorem{remark}{Remark}
\theoremstyle{definition}
\newtheorem{definition}{Definition}
\newtheorem{example}{Example}

\usepackage{csquotes}

\usepackage{algorithmic}
\usepackage[noend,lined,linesnumbered]{algorithm2e} 

\usepackage{longtable}

\ifpets
\newcommand{\descr}[1]{\vspace{0.2cm} \noindent \textbf{\sffamily #1}}
\else
\newcommand{\descr}[1]{\vspace{0.2cm} \noindent \textbf{#1}}
\fi

\usepackage{enumitem} 

\usepackage{xcolor}
\definecolor{teal}{HTML}{008080}

\usepackage{hyperref}
\hypersetup{
   colorlinks=true,
   citecolor=teal
}

\usepackage{multirow}

\allowdisplaybreaks[1] 

\begin{document}

\ifpets
\else

\title{Averaging Attacks on Bounded Noise-based Disclosure Control Algorithms\thanks{This is a preliminary version of the paper with the same title accepted for publication in the Proceedings on the 20th Privacy Enhanching Technologies Symposium (PETS 2020).}}
\author{Hassan Jameel Asghar \& Dali Kaafar \\
\small Macquarie University, Australia\\
\small \texttt{\{hassan.asghar, dali.kaafar\}@mq.edu.au}
}
\date{}
\fi

\maketitle
\ifpets
\else

\begin{abstract}
{We describe and evaluate an attack that reconstructs the histogram of any target attribute of a sensitive dataset which can only be queried through a specific class of real-world privacy-preserving algorithms which we call \emph{bounded perturbation} algorithms. A defining property of such an algorithm is that it perturbs answers to the queries by adding zero-mean noise distributed within a bounded (possibly undisclosed) range. Other key properties of the algorithm include only allowing restricted queries (enforced via an online interface), suppressing answers to queries which are only satisfied by a small group of individuals (e.g., by returning a zero as an answer), and adding the same perturbation to two queries which are satisfied by the same set of individuals (to thwart differencing or averaging attacks). A real-world example of such an algorithm is the one deployed by the Australian Bureau of Statistics' (ABS) online tool called TableBuilder, which allows users to create tables, graphs and maps of Australian census data~\cite{tbe-algo}. We assume an attacker (say, a curious analyst) who is given oracle access to the algorithm via an interface. We describe two attacks on the algorithm. Both attacks are based on carefully constructing (different) queries that evaluate to the same answer. 
The first attack finds the hidden perturbation parameter $r$ (if it is assumed not to be public knowledge). The second attack removes the noise to obtain the original answer of some (counting) query of choice. We also show how to use this attack to find the number of individuals in the dataset with a target attribute value $a$ of any attribute $A$, and then for all attribute values $a_i \in A$. None of the attacks presented here depend on any background information. Our attacks are a practical illustration of the (informal) fundamental law of information recovery which states that ``overly accurate estimates of too many statistics completely destroys privacy''~\cite{dinur-nissim, conc-dp}.} 
\end{abstract}

\tableofcontents

\fi
\section{Introduction}
\label{sec:intro}
We consider online privacy-preserving algorithms
that return \emph{noisy} answers to queries on sensitive data, where the {zero-mean} noise is strictly bounded between an interval parameterised by a \emph{perturbation parameter}. Our focus is restricted to algorithms that (privately) answer \emph{counting} queries. An example counting query is: ``How many people in the dataset are aged 25 and live in the suburb of Redfern in New South Wales, Australia?'' An example of such privacy-preserving algorithms is the perturbation algorithm employed by the TableBuilder tool from the Australian Bureau of Statistics (ABS), which allows access to the Australian population census data.\footnote{See \url{http://www.abs.gov.au/websitedbs/censushome.nsf/home/tablebuilder}. At the time of this writing, there are three flavours of TableBuilder. The first two are TableBuilder Basic \& Pro, which require registration (the latter being a charged product). After the registration request is approved, the user can login to use the TableBuilder tool. The third flavour is for guests, called TableBuilder Guest. This can be accessed by users without registration and provides access to fewer variables from the census data. The attacks mentioned in this paper are applicable to all flavours.} We shall call this algorithm the TBE algorithm named after its authors~\cite{tbe-algo}. The TBE algorithm and similar \emph{bounded perturbation algorithms} are built on certain principles to address privacy and utility concerns, outlined below. 
\begin{itemize}
 	\item Access to sensitive data is only allowed through a \emph{restricted} query interface. This limits the types of queries that can be executed via the underlying (privacy-preserving) algorithm, therefore minimizing information leakage by ensuring that the (effective) query language is not rich enough. A rich query language would require query auditing to ensure privacy; such auditing may not even be programmable~\cite{dp-book}. 
	\item The noise added to the queries is bounded within a predetermined range, say $\pm 3$ of the actual answer. From a privacy angle this adds uncertainty if the (adversarial) analyst is trying to run a query on certain attributes in the dataset to infer some information about a target individual. From a utility point of view, the bounded noise ensures that the noise never overwhelms the true statistics. 
	\item The algorithm suppresses low non-zero counts (e.g., by returning 0). This makes it hard for an analyst to know if certain characteristics (combination of attributes or fields in the dataset) are shown by its target individual(s) or not. For instance, a 0 count could be an actual 0 or a 1 in the original dataset.
	\item The algorithm adds the exact same noise if the answers returned by two queries are contributed by the same set of \emph{contributors}. A contributor to a query is any individual that satisfies the query. This is a defence against averaging attacks~\cite{fraser-differencing}, where the analyst cannot pose multiple, possibly differently structured, queries with the same set of contributors to reduce noise by averaging to find the true count.
\end{itemize}

\descr{Contributions.} We show an attack that retrieves the entire histogram of a target attribute from a dataset which can only be queried through the TBE algorithm.\footnote{An example of an attribute is `Age', and its histogram is the number of people of each age.} Our attack relies on carefully constructing queries that yield the same (true) answer and averaging them over all queries to eliminate noise. Furthermore, in cases where it may be argued that the perturbation parameter is not public information, we show an attack that retrieves the exact (hidden) perturbation parameter. We remark that the attacks presented do not depend on any background knowledge about individuals in the dataset, i.e., they are dataset independent, and hence applicable to any underlying dataset.\footnote{Barring a few mild assumptions on the domain of the dataset, e.g., the existence of an attribute with more than 2 attribute values (Section~\ref{sub:find-perturb}).} We discuss several mitigation measures, and argue that the most sound strategy is to add noise as a function of the number of queries. This follows from the bound on the success probability of our attack, and is consistent with the amount of noise required via the notion of differential privacy~\cite{calib-noise}. 

\descr{Results.} For both attacks, i.e., finding the hidden perturbation and removing noise, we derive exact expressions for the success probabilities as a function of the perturbation parameter and the number of queries to the algorithm. We also evaluate the noise removing attack on a synthetic dataset queried through an API to the TBE algorithm. Our results (both theoretical and experimental) show that any perturbation parameter less than or equal to 10 can be retrieved with probability $\approx 0.90$ with only up to 1,000 queries. Furthermore, we are able to recover a smaller perturbation parameter ($5$), which is desirable for utility, with only 200 queries with a probability of more than $0.95$. Using the same API, with the perturbation parameter $2$, we retrieve an entire histogram of a target column of the synthetic dataset with more than 107 attribute values through only 400 queries per attribute value (via the noise removing attack). The attack also successfully retrieves suppressed counts (low counts returned as 0), and hence distinguishes between actual zeros and suppressed zeros. 

\descr{Application to the ABS TableBuilder.} Our use of the API to query the TBE algorithm simulates the setting of the ABS TableBuilder tool providing access to the Australian census data. The TableBuilder tool does not currently have a programmable API, and can only be accessed via a web interface. The attack in practice can still be launched by either manually querying TableBuilder to construct tables or more realistically, by crafting web queries through scripts to directly query the JavaScript programs behind the web interface. We chose to use the simulated setting for a quicker illustration of the attack and more importantly due to ethical considerations; the census data being highly sensitive. 

\descr{Privacy Implications of Our Attacks.} Our main attack removes noise in the answers returned by the TBE algorithm. This is specifically problematic for low counts, e.g., counts of 1. For instance, assume that there is a single individual in the dataset gendered male and within age bracket 30-39 who lives in the suburb Newtown. Since the true answer is 1, TBE will return the \emph{suppressed} answer 0. Our attack enables the analyst to retrieve the true count 1 by cleverly constructing queries that return larger counts (cf. Section~\ref{subsub:broaden}). Once the true count is revealed, the analyst having the \emph{background knowledge} (male, 30-39, Newtown) can successfully re-identify the person in the dataset. Thus, true counts enable other privacy attacks such as re-identification and inference (cf. Section~\ref{sub:priv-attacks}). Note that it is to avoid such attacks that the TBE algorithms employs the aforementioned principles to hide true counts. We remark that some international government agencies such as Statistics Sweden have expressed interest in the use of TableBuilder for disclosure control of frequency tables~\cite{keefe, stat-sweden}, although it is acknowledged that further evaluation of the technique is necessary~\cite{stat-sweden}. Moreover, there are plans to expand the use of TableBuilder to other Australian national government agencies and datasets.
Thus our results have implications beyond the ABS use of TableBuilder.

\section{Preliminaries}
\label{sec:prelim}
We model the database $D$ as a set of rows of data, each belonging to a unique individual from a finite set of individuals $U$. Thus, the size of the dataset is the same as the size of the set $U$, i.e., $|D| = |U|$. The data from an individual $u \in U$ is represented as a row $x \in D$. We denote the link by $x = \text{data}(u)$. The row $x$ is a member of some domain $\mathbb{D}$. 

\subsection{Definitions: Queries and Contributors}
\begin{definition}[Query]
\label{def:query}
A query $q: \mathbb{D} \rightarrow \{0, 1\}$ is defined as a predicate on the rows $x \in D$. Note that this is in fact the definition of a \emph{counting query}. The queries in this document are restricted to counting queries. The query's result on the dataset $D$ is defined as $q(D) = \sum_{x \in D} q(x)$. For any two queries $q_1$ and $q_2$, we denote by $q_1 \wedge q_2$ the predicate that evaluates to $1$ on a row if and only if both $q_1$ and $q_2$ evaluate to 1 on the row. Likewise we denote by $q_1 \vee q_2$ the predicate that evaluates to $1$ if either $q_1$ or $q_2$, or both evaluate to $1$. \qed
\end{definition}
We will often omit the argument of $q$, i.e., $D$, since we are concerned with a single dataset in this document.

\begin{definition}[Contributors]
A \emph{contributor} of a query $q$ is any individual $u \in U$ such that $q(x) = 1$, where $x = \text{data}(u)$. The \emph{set of contributors} of a query $q$, denoted $C(q)$ is defined as
\[
C(q) = \left\{ u \in U : q(x) = 1, \text{ where } x = \text{data}(u) \right\}
\] 
Two queries $q_1$ and $q_2$ are said to have the same contributors if $C(q_1) = C(q_2)$. Otherwise they are said to have different contributors. \qed
\end{definition}
Note that having different contributors does not mean that $C(q_1)$ and $C(q_2)$ are necessarily disjoint. Furthermore, it is possible for two different queries (different predicates) $q_1$ and $q_2$ to have the same contributors (depending on the dataset). 
We assume the dataset $D$ to be vertically divided into \emph{attributes}. Let $A$ denote one such attribute, and let $|A|$ denote its cardinality, i.e., the number of \emph{attribute values} of $A$. Let $a \in A$ be an attribute value. We assume that the data of each $u \in U$ takes on only one value from $A$. The query $q_a$ is defined as the predicate which evaluates to $1$ if the row has value $a$ under $A$. Let $A' \subseteq A$, then $q_{A'}$ is defined as $q_{A'} = \vee_{a \in A'} (q_{a})$. We shall call this query, the \emph{total query}, as it returns the total number of counts that satisfy each of the attribute values $a \in A'$. We also denote the trivial query $q_\emptyset$, which evaluates to 1 on every row. Hence $q_\emptyset(D) = n$. Also, note that $q_A(D) = n$ for every attribute $A$ of $D$. Clearly, in both cases the set of contributors is the entire user set $U$. 

\begin{table}[!h]
\centering
\caption{An example database $D$ with $|U| = 6$ users.}
\label{table:d-example}
\begin{tabular}{l|c|c|c}
$U$ & Suburb & Age & Gender\\
\hline\hline
$1$ & Redfern & 20-29 & M \\
$2$ & Redfern & 20-29 & M \\
$3$ & Newtown & 30-39 & F \\
$4$ & Redfern & 20-29 & F \\
$5$ & Surry Hills & 40-49 & M \\
$6$ & Darlinghurst & 70-79 & F\\
\end{tabular}
\end{table}

\begin{example}
\label{ex:queries}
Table~\ref{table:d-example} shows a dataset $D$ with 6 users. We have three different attributes: Suburb, Age and Gender. The attribute $U$ is just shown for illustrative purposes. It is otherwise forbidden to be queried. The attribute Suburb has 4 attribute values, Age has 5 attribute values and Gender has 2 attribute values. Queries $q_{\text{Redfern}}$ and $q_{\text{20-29}}$ both evaluate to $3$. Also note that $C(q_{\text{Redfern}}) = C(q_{\text{20-29}}) = \{1, 2, 4\}$, and thus the two queries have the same set of contributors. On the other hand, $C(q_{\text{Redfern}}) \neq C(q_{\text{M}})$. We have $(q_{\text{Redfern}} \wedge q_{\text{M}})(D) = 2$, and $(q_{\text{Redfern}} \vee q_{\text{M}})(D) = 4$. Let $A' = \{\text{Redfern}, \text{Newtown}\}$. Then the total query $q_{A'}(D)$ has answer $4$. \qed
\end{example}

\subsection{Background on Privacy Attacks}
\label{sub:priv-attacks}
We briefly describe some known categories of privacy attacks. 
For ease, we assume a database $D$ with $n$ users and three attributes $U$, $A$ and $B$. As in Table~\ref{table:d-example}, $U$ represented unique user identities. We can further assume $A$ to be the attribute Suburb, and $B$ to be the attribute Age in the database of this table. Queries $q$ are defined on $D$ as in the previous section. We assume that the database can be queried via a mechanism $\mathcal{M}$ only, which is possibly randomized. To define the various attacks, we shall take example mechanisms $\mathcal{M}$.

We first assume an adversary who has some \emph{background knowledge} of some user $u \in U$. In particular, the adversary knows that the user has some attribute value $a \in A$. Consider a mechanism $\mathcal{M}$ which simply strips the attribute $U$ and releases exact answers to queries on attribute $A$ and/or $B$. Suppose $q_a(D) = 1$, i.e., only one individual takes on $a \in A$ in the dataset $D$. Then this mechanism is susceptible to a \emph{re-identification attack}. The adversary knowing $a \in A$ has re-identified $u$'s data in the database. The adversary can further launch an \emph{inference attack} by asking the queries $q_a \wedge q_b$ for all $b \in B$, thus inferring what value is taken by the user $u$ in the attribute $B$. It is usually argued that for re-identification to be successful the attribute $a$ should also be unique in the population, and not just in the dataset. However, if the attribute $A$ is replaced with a large enough set of attributes as background knowledge, then the resulting attribute-value tuple is likely to be unique in the population as well~\cite{sweeney}. Thus, any mechanism that returns exact answers is susceptible to re-identification attacks. We also note that the success of the inference attack is not necessarily tied to successful re-identification. Indeed, $q_a(D)$ might be greater than $1$, but the adversary can still ask the queries $q_a \wedge q_b$ for all $b \in B$ and learn which values $b$ are \emph{not} taken by its target.

To mitigate re-identification, we can redefine the mechanism to suppress low counts: any query answer less than a suppression parameter $s \ge 0$ is clipped to 0. However, this mechanism is susceptible to a \emph{differencing attack}. The adversary defines the subset $A' = A - \{a\}$, asks the queries $q_A$ and $q_{A'}$, and subtracts the answer to the second query from the first. If both $q_A(D)$ and $q_{A'}(D)$ are greater than $s$, then the difference in the two answers reveals $q_a(D)$. To thwart the differencing attack, we can design a mechanism $\mathcal{M}$ which, instead of suppression, perturbs all answers by adding fresh random noise from the set $\{-1, 0, 1\}$. In this case, the attacker can launch an \emph{averaging attack}. In particular, the attacker asks for the answer of $q_a$ a total of $t$ times, and then averages the answers. With increasing $t$, the probability that the average deviates from the true answer approaches 0. A solution is to add the same noise if the same query is asked again, or if the same contributors satisfy the query, as is done, for instance, in the TBE algorithm. 

Yet another form of attack is \emph{database reconstruction}, where the adversary reconstructs a target column of a dataset, e.g., corresponding to attribute $A$. In other words, the adversary attempts to exactly retrieve the values taken by each user under $A$ in the dataset. The database reconstruction attacks in the literature~\cite{exposed, dinur-nissim, dwork-recon, power-recon, cohen-linear} require some form of queries that can select rows corresponding to different subsets of users $U$. An example of such queries are subset sum queries. Previous work shows that any mechanism that returns noisy answers to subset sum queries where the noise scale is bounded by a constant is susceptible to database reconstruction attacks. Alternatively, the attack can be launched by using a set of attributes (instead of $U$) such that each user in the dataset takes a unique attribute-value tuple in the dataset. See~\cite[\S 3]{exposed} for more details. 

Our main proposed attack is an averaging attack, which in combination with a differencing attack, results in a \emph{histogram reconstruction attack}. Namely, For any given $a \in A$, we construct exact answers to the queries $q_{a} \wedge q_b$, for all $b \in B$. It is not clear how a database reconstruction attack can be launched on the TBE algorithm, as its query interface is very restrictive. 


%
%

\section{Privacy Algorithms}
\label{sec:algos}
Our focus is on a particular privacy algorithm (mechanism) that returns (perturbed) answers to queries $q$ on the database $D$, where the queries are as defined in Definition~\ref{def:query}. We call the algorithm the \emph{Bounded Noisy Counts} algorithm. The algorithm returns the answer to a query $q$ by adding bounded noise $e$ from the uniform distribution over the set of integers in the interval $[-r, r]$, where $r$ is some positive integer, i.e., the perturbation parameter. We shall denote the set of integers in $[-r, r]$ by $\mathbb{Z}_{\pm r}$ and the (discrete) uniform distribution over $\mathbb{Z}_{\pm r}$ by $\mathbb{U}_{\pm r}$. The algorithm also has two exceptional cases:
\begin{enumerate}
	\item If the answer to $q$ is less than a suppression parameter $s \ge r$, then the answer returned is exactly 0. 
	\item If two queries $q_1$ and $q_2$ have the same contributors, then the noise $e$ added to the two queries is the same. 
\end{enumerate}
The algorithm therefore is a stateful algorithm where the state consists of a dictionary of subsets of contributors and the corresponding noise. We denote this algorithm by $\mathcal{M}_{r, s}$ and on any input query $q$ denote its output as $\mathcal{M}_{r, s}(q)$. The algorithm is described in Algorithm~\ref{algo:bounded-noise}. We remark that the noise distribution can be any \emph{admissible distribution}, which we define in Section~\ref{sec:entropy}. Bounded uniform random noise is one example.

\begin{algorithm}[!ht]
\SetAlgoLined
\SetAlCapSkip{1em}
\DontPrintSemicolon{}
\let\oldnl\nl
\newcommand{\nonl}{\renewcommand{\nl}{\let\nl\oldnl}}
\SetKwInOut{Input}{Input}
\SetKwInOut{State}{State}
\Input{The query $q$, perturbation parameter $r$, suppression parameter $s \ge r$.}
\State{A noise dictionary, denoted $\mathsf{nd}$, with keys from subsets of $U$ and values in $\mathbb{Z}_{\pm r}$.}
Evaluate $C(q)$ and let $n = q(D)$.\;
\eIf {$n \le s$}{
	return 0.\;
	}{
	\eIf {$C(q) \in \mathsf{nd}$}{
		obtain noise $e \leftarrow \mathsf{nd}(C(q))$.\;
		return $n + e$.\;
		}{
		sample $e \sim \mathbb{U}_{\pm r}$.\;
		add entry $\mathsf{nd}(C(q)) = e$.\;
		return $n + e$.\;
		}
	}
\caption{The Bounded Noisy Counts Algorithm $\mathcal{M}_{r, s}$.}
\label{algo:bounded-noise}
\end{algorithm}

From now on, we shall drop the subscripts $r$ and $s$, and denote the algorithm simply as $\mathcal{M}$. Following properties are direct consequences of the algorithm.
\begin{proposition}
\label{prop:bnc}
Let $\alpha \leftarrow \mathcal{M}(q)$ be the answer returned by $\mathcal{M}$ on some query $q$. Let $n = q(D)$. Then 
\begin{enumerate}[label=(\alph*)]
\item $\alpha \ge 0$,
\item $n - s \le \alpha \le n + r$.
\item If $\alpha > 0$ then $C(q) \neq \emptyset$.
\end{enumerate} 
\end{proposition}
\begin{proof}
See Appendix~\ref{app:proofs}.
\end{proof}

The above algorithm is used as a subroutine by another algorithm which we call the \emph{Attribute Analyser}. This algorithm helps the querier analyse multiple attribute values under an attribute at once. Let $\mathbf{B}$ denote a tuple of (one or more) attributes in $D$. Let $\mathbf{b} \in \mathbf{B}$ denote the vector of attribute values whose $i$th entry is one of the attribute values of the $i$th attribute in $\mathbf{B}$. Let $q_{\mathbf{b}}$ denote the predicate that evaluates to 1 if and only if the row satisfies all values in $\mathbf{b}$. This models a target sub-population in the dataset $D$. The Attribute Analyser takes an attribute value vector $\mathbf{b} \in \mathbf{B}$, and a subset of attribute values $A' \subseteq A$, where $A \notin \mathbf{B}$. Let $|A'| = m$. The algorithm then runs $\mathcal{M}$ on (i) each of the queries $q_\mathbf{b} \wedge q_{a_i}$ where $a_i \in A', i \in \{1, \ldots, m\}$ obtaining answers $\alpha_i$, and (ii) on the \emph{total} query $q_\mathbf{b} \wedge q_{A'}$, obtaining the answer $\alpha_{A'}$. It then returns the answer vector $(\alpha_1, \ldots, \alpha_{m}, \alpha_{A'})$. 

\begin{algorithm}[!ht]
\SetAlgoLined
\SetAlCapSkip{1em}
\DontPrintSemicolon{}
\let\oldnl\nl
\newcommand{\nonl}{\renewcommand{\nl}{\let\nl\oldnl}}
\SetKwInOut{Input}{Input}
\SetKwInOut{State}{State}
\Input{Attribute value vector $\mathbf{b} \in \mathbf{B}$, attribute subset $A' \subseteq A$ of cardinality $m$ (where $A \notin \mathbf{B}$).}
\For{$i = 1$ \KwTo $m$}{
	Let $a_i$ be the $i$th element in $A'$.\;
	Obtain $\alpha_i \leftarrow \mathcal{M}(q_\mathbf{b} \wedge q_{a_i})$.\;
}
%
Obtain $\alpha_{A'} \leftarrow \mathcal{M}(q_\mathbf{b} \wedge q_{A'})$.\;
return $(\alpha_1, \ldots, \alpha_m, \alpha_{A'})$.\; 
\caption{The Attribute Analyser Algorithm.}
\label{algo:att-analyze}
\end{algorithm}

Note that $A'$ can be possibly empty, in which case $q_{A'} = q_\emptyset$, and the algorithm returns the answer to $q_\mathbf{b} \wedge q_\emptyset = q_\mathbf{b}$ only. Likewise, $\mathbf{B}$ can be possibly empty, meaning that $q_\mathbf{b} = q_\emptyset$, in which case we are analysing $A'$ over the whole dataset $D$ (rather than over a sub-population).

\begin{example}
\label{ex:att-analyse}
Consider the dataset from Table~\ref{table:d-example}. Let $\mathbf{B} = (\text{Suburb}, \text{Gender})$. Furthermore, let $\mathbf{b} \in \mathbf{B}$ be $(\text{Redfern}, \text{M})$. Thus, we are interested in the sub-population of people who are male and living in the suburb of Redfern in the dataset. Thus $q_\mathbf{b} = q_{\text{Redfern}} \wedge q_{\text{M}}$. Let $A = \text{Age}$, and $A' \subseteq A$ be $\{\text{20-29}, \text{30-39}\}$. Then $\alpha_1$ corresponds to $q_\mathbf{b} \wedge q_{\text{20-29}}$ (true answer 2), $\alpha_2$ corresponds to $q_\mathbf{b} \wedge q_{\text{30-39}}$ (true answer 0), and $\alpha_{A'}$ corresponds to $q_\mathbf{b} \wedge q_{A'} = (q_\mathbf{b} \wedge q_{\text{20-29}}) \vee (q_\mathbf{b} \wedge q_{\text{30-39}})$ (true answer 2). If on the other hand, we have $B = \text{Suburb}$ and $b = \text{Redfern}$, $A = \text{Gender}$ and $A' = \{ \text{M}, \text{F} \}$, then we get $(q_b \wedge q_\text{M})(D) = 2$, $(q_b \wedge q_\text{F})(D) = 1$, and $(q_b \wedge q_{A'})(D) = 3$. Note that $C(q_b \wedge q_\text{M}) \neq C(q_b \wedge q_\text{F}) \neq C(q_b \wedge q_{A'})$. Thus, $\mathcal{M}$ would add fresh noise values to each of these true counts. If the suppression parameter $s$ is set to $1$, then the answer to $q_b \wedge q_\text{F}$ would be fixed to 0.\qed 
\end{example}

\section{Privacy Attacks}
\label{sec:attacks}
We assume an attacker (say, a curious analyst) who is given oracle access to the Attribute Analyser (which in turn uses the Bounded Noisy Counts algorithm as a subroutine). The attacker does not know the parameter $r$. We describe two attacks on the algorithm. The first attack finds the hidden perturbation parameter $r$. The second attack removes the noise to obtain the original count $n = q(D)$ of some query of choice $q$. We also show how to use this attack to obtain the value $q_a(D)$ for some target attribute value $a \in A$, and then for all attribute values $a \in A$. We reiterate that none of the attacks depend on any background information. For simplicity, we assume that $s = r$. Our fundamental unit of measurement will be the number of queries $t$ submitted to the Bounded Noisy Counts algorithm.

\subsection{Attack 1: Finding the Perturbation Parameter \texorpdfstring{$r$}{r}}
\label{sub:find-perturb}
Let $b \in B$ be an attribute value and let $A \neq B$ be an attribute with only two attribute values $a_1$ and $a_2$, e.g., the Gender attribute with values male and female. Let $n = q_b(D)$, $n_1 = (q_b \wedge q_{a_1})(D)$ and $n_2 = (q_b \wedge q_{a_2})(D)$.
Consider the sequence of inputs $(b, \{a_1\})$, $(b, \{a_2\})$ and $(b, \emptyset)$ to the Attribute Analyser. As output, we obtain $n_1 + e_1$, $n_2 + e_2$ and $n + e_3$, where $e_i$ are the noise terms added by Bounded Noisy Counts. Clearly, $n = n_1 + n_2$. Furthermore, 
\begin{lemma}
\label{lem:e:ind}
If $n_1, n_2 > r$, then $e_1$, $e_2$ and $e_3$ are independent samples from the distribution $\mathbb{U}_{\pm r}$.
\end{lemma}
\begin{proof}
See Appendix~\ref{app:proofs}.
\end{proof}
\begin{lemma}
\label{lem:diff:cont}
Let $b_1, \ldots, b_m$ be different attribute values from one or more attributes. 
If $\mathcal{M}(q_{b_i}) \neq \mathcal{M}(q_{b_j})$ then $C(q_{b_i}) \neq C(q_{b_j})$, for all $i, j \in [m]$, $i \ne j$.
\end{lemma}
\begin{proof}
See Appendix~\ref{app:proofs}.
\end{proof}
Now, define the random variable
\begin{equation}
Z = (n_1 + E_1) + (n_2 + E_2) - (n + E_3) = E_1 + E_2 - E_3 \label{eq:rem-err}
\end{equation}
where $E_i$ are i.i.d. random variables with distribution $\mathbb{U}_{\pm r}$. Since $E_i \le r$, we have $Z \le 3r$, which would happen if $E_1 = E_2 = r$ and $E_3 = 
-r$. Our attack can be summarised as follows: 
\begin{enumerate}
	\item Find an attribute $A$ with only two attribute values $a_1$ and $a_2$ (e.g., gender).
	\item Find $m$ different attribute values $b_1, \ldots, b_m$ from \emph{any} number of attributes such that $\mathcal{M}(q_{b_i} \wedge q_{a_1})$ and $\mathcal{M}(q_{b_i} \wedge q_{a_2})$ are greater than $0$ implying that $\mathcal{M}(q_{b_i}) > 0$. This ensures the condition of Lemma~\ref{lem:e:ind}. Furthermore, ensure that the contributors of all queries $q_{b_i}$ are different. Lemma~\ref{lem:diff:cont} shows how to ensure this. 
	\item For the $i$th attribute value ($b_i$) obtain $z_i$ which is an instance of the random variable $Z$ in Eq.~\ref{eq:rem-err}.
	\item Let $z_{\max}$ be the maximum of the $m$ values. We then return $\lceil \frac{z_{\max}}{3} \rceil$ as the guess for $r$. 
\end{enumerate}
We can in fact do better by also keeping track of the minimum values. Let $z_{\min}$ be the minimum of the $m$ 
values. Notice that $Z \ge -3r$. Our guess for $r$ is then $\max\{ - \lceil \frac{z_{\min}}{3} 
\rceil, \lceil \frac{z_{\max}}{3} \rceil \}$. The guess for $r$ would be correct as long as either $-z_{\min}$ or $z_{\max}$ is greater than $3(r - 1)$. The exact algorithm is described in Algorithm~\ref{algo:pert-finder}. 
\begin{algorithm}[!ht]
\SetAlgoLined
\SetAlCapSkip{1em}
\DontPrintSemicolon{}
\let\oldnl\nl
\newcommand{\nonl}{\renewcommand{\nl}{\let\nl\oldnl}}
\SetKwInOut{Input}{Input}
\SetKwInOut{State}{State}
\Input{$m$ distinct attribute values $b_1, \ldots, b_m$, attribute $A$ of cardinality $2$ with attributes $a_1$ and $a_2$, all satisfying $\mathcal{M}(q_{b_i} \wedge q_{a_1}), \mathcal{M}(q_{b_i} \wedge q_{a_2}) > 0$ and $\mathcal{M}(q_{b_i}) \neq \mathcal{M}(q_{b_j})$, for $i, j \in [m]$, $i \neq j$.}
Set $z_{\min} \leftarrow \infty$ and $z_{\max} \leftarrow -\infty$.\;
\For{$i = 1$ \KwTo $m$}{
	Run Attribute Analyzer with inputs $(b_i, \{a_1\})$, $(b_i, \{a_2\})$ and $(b_i, \emptyset)$ and get outputs $z_1$, $z_2$ and $z_3$, respectively.\;
	Set $z = z_1 + z_2 - z_3$.\;
	\If {$z > z_{\max}$}{
		$z_{\max} \leftarrow z$.\;
	}
	\If {$z < z_{\min}$}{
		$z_{\min} \leftarrow z$.\;
	}
}
Let $r' = \max\{ - \lceil \frac{z_{\min}}{3} \rceil, \lceil \frac{z_{\max}}{3} \rceil \}$.\;
Output $r'$.\;
\caption{The Perturbation Finder Algorithm.}
\label{algo:pert-finder}
\end{algorithm}

We will show that the algorithm returns the correct perturbation $r$ with high probability, depending on a suitable choice for $m$. 
\begin{lemma}
\label{lem:twenty}
Let $r \ge 1$, and let $E_1$, $E_2$ and $E_3$ be variables that take values in $\mathbb{Z}_{\pm r}$. Out of the $(2r+1)^3$ possible values of the tuple $(E_1, E_2, E_3)$, there are precisely 20 that satisfy $E_1 + E_2 + E_3 > 3(r-1)$ or $E_1 + E_2 + E_3 < -3(r-1)$.
\end{lemma}
\begin{proof}
See Appendix~\ref{app:proofs}.
\end{proof}

\begin{proposition}
\label{prop:pertsucc}
Let $r'$ be the output of Perturbation Finder. Then
\[
\Pr [ r' = r ] = 1 - \left(1 - \frac{20}{(2r+1)^3} \right)^m
\] 
\end{proposition}
\begin{proof}
See Appendix~\ref{app:proofs}.   
\end{proof}

For a given probability of success, larger perturbations require more queries to $\mathcal{M}$ (through the Perturbation Finder algorithm). However, note that larger values of $r$ are not desirable from a utility point of view. For each attribute value, the Attribute Analyser makes 3 calls to Bounded Noisy Counts $\mathcal{M}$. Thus, for a total of $m$ attributes we have $t = 3m$ queries to $\mathcal{M}$. Figure~\ref{fig:pertsucc} shows the number of queries $t$ required for a given probability of success. Note that smaller values of $r$, i.e., $\le 5$, which are desirable from a utility point of view need less than $t = 600$ for a 95\% success rate. 

\begin{figure}[ht!]
\centering
\ifpets
\else
\includegraphics[scale=0.8]{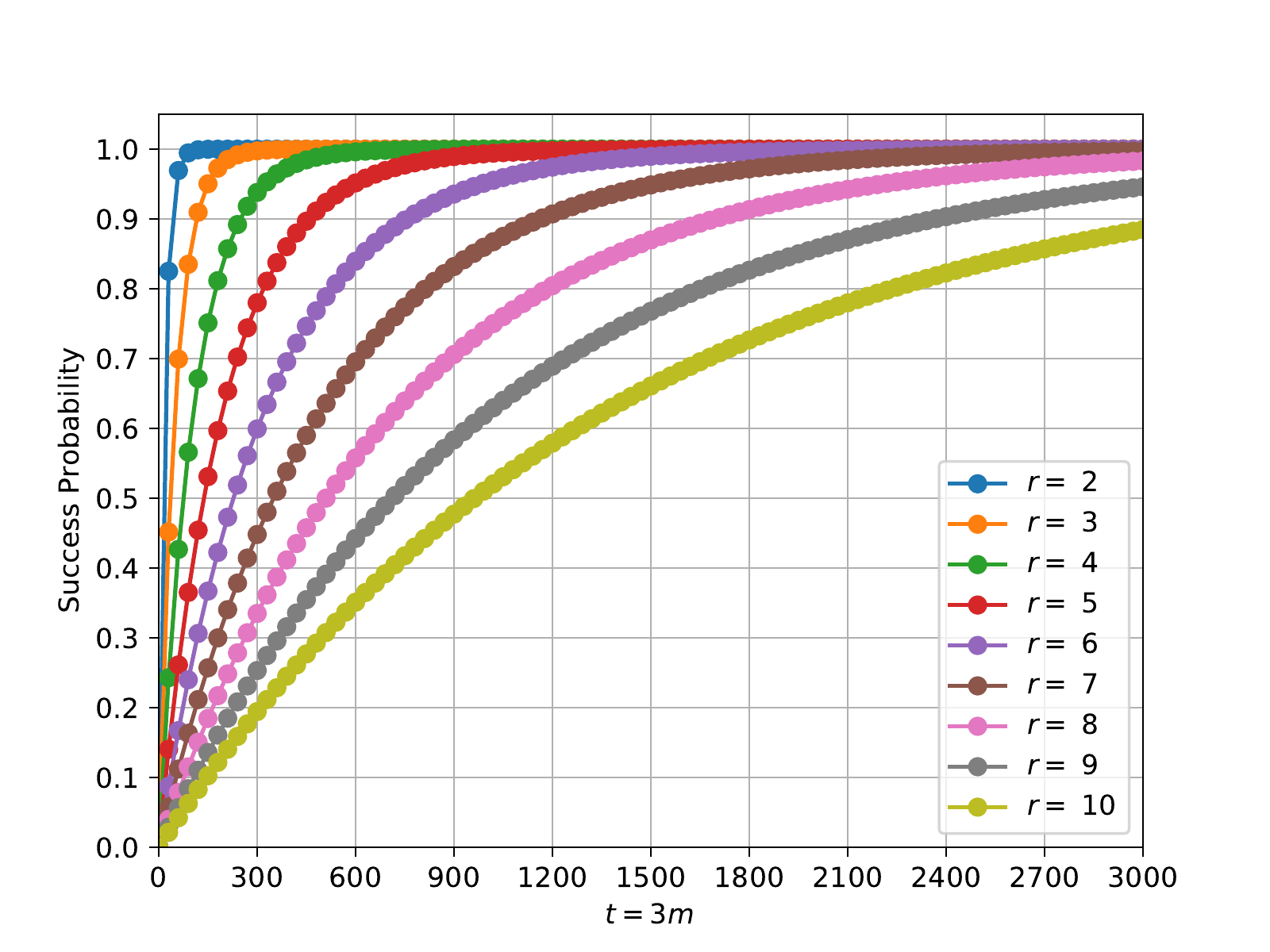}
\fi
\caption{Probability of successfully finding the perturbation parameter $r$ as a function of the number of queries $t$ to Bounded Noisy Counts in the Perturbation Finder algorithm.  Higher perturbations require a much larger number of attribute values $m$ (and hence number of queries).}
\label{fig:pertsucc}
\end{figure}

\begin{remark}
We have assumed for simplicity that $A$ is an attribute with exactly two attribute values. In general, the attack is applicable to any attribute with \emph{at least} two attribute values. In this case, we run the Attribute Analyser with inputs corresponding to the two selected attribute values, plus the input $(b_i, A')$.
\end{remark}

\begin{remark}
Again for simplicity, in the $i$th iteration of Algorithm~\ref{algo:pert-finder}, we run the Attribute Analyser on three different inputs. These can be replaced by a single input $(b_i, A')$, where $A' = \{a_1, a_2\} \subseteq A$. The output of Attribute Analyser will by definition return (noisy) answers to the queries $q_b \wedge q_{a_1}$, $q_b \wedge q_{a_2}$ and $q_b \wedge q_{A'}$ as desired (step 5 of Algorithm~\ref{algo:att-analyze}). Thus, while this constitutes 3 queries to Bounded Noisy Counts, this is only a single query to the Attribute Analyser. The latter resembles the TableBuilder tool interface. The results summarised in Section~\ref{sec:intro} are with respect to the Attribute Analyser, which simulates the TableBuilder interface.
\end{remark}


\subsection{Attack 2: Removing Noise}
Consider again a tuple of attributes $\mathbf{B}$ (possibly empty) from $D$, and let $\mathbf{b} \in \mathbf{B}$ denote a vector of attribute values from $\mathbf{B}$ defined as before. Consider an attribute $A$ with $m$ attribute values $a_1, \ldots, a_m$. In this section, we will show an attack that finds the exact answer to $q_{\mathbf{b}} \wedge q_A$ by using the Attribute Analyzer as a black box. We will then show how to use this algorithm to find the true answer to any \emph{target} query $q_{\mathbf{b}} \wedge q_{a_i}$. Continuing on, we can find the true answers to all queries $q_{\mathbf{b}} \wedge q_{a_i}$, $i \in [m]$. We will first assume that $\mathcal{M}(q_{\mathbf{b}} \wedge q_{a_i}) \neq 0$ for all $i$, for simplicity. Later on we will show that this assumption can be relaxed as long as we have some $m' < m$ attributes from $A$ satisfying $\mathcal{M}(q_{\mathbf{b}} \wedge q_{a_i}) \neq 0$. 

We begin with a simple observation on $A$. Recall that a two-partition of a set $A$ is a partition of $A$ with exactly two subsets of $A$. 
\begin{lemma}
\label{lem:numtwoparts}
There are exactly $2^{m - 1} - 1$ two-partitions of the set $A$.  
\end{lemma}
\begin{proof}
See Appendix~\ref{app:proofs}.
\end{proof}
We will let $P_A$ denote the set of all two-partitions of $A$. The following result shows that all sets in $P_A$ have different contributors.
\begin{lemma}
\label{lem:twopart}
Assume $\mathcal{M}(q_\mathbf{b} \wedge q_{a_i}) \ne 0$, for all $i \in [m]$. Let $\mathcal{A}$ be a two-partition in $P_A$, and let $A' \in \mathcal{A}$ be any of the two sets in $\mathcal{A}$. Let $A''$ be either the other set in $\mathcal{A}$ or any of the two sets from any other partition in $P_A$. Then, $C(q_\mathbf{b} \wedge q_{A'}) \ne C(q_b \wedge q_{A''})$.
\end{lemma}
\begin{proof}
See Appendix~\ref{app:proofs}.
\end{proof}

\begin{example}
\label{ex:part}
Consider the dataset in Table~\ref{table:d-example}. Let $A = \text{Suburb}$. Then $A$ has $m = 4$ attribute values: D = Darlinghurst, N = Newtown, R = Redfern and S = Surry Hills. The $2^{m - 1} - 1 = 7$ two-partitions of $A$ are as follows:

\ifpets
\else
\begin{center}
\begin{tabular}{c|c|c|c|c|c|c|c} 
$A_1$ & \{D\} & \{N\} & \{R\} & \{S\} & \{D, N\} & \{D, R\} & \{D, S\} \\
$A_2$ & \{N, R, S\} & \{D, R, S\} & \{D, N, S\} & \{D, N, R\} & \{R, S\} & \{N, S\} & \{N, R\} \\
\end{tabular}%
\end{center}
\fi


Let $\mathbf{B} = \emptyset$ and hence $\mathbf{b} = \emptyset$. Then $q_{\mathbf{b}} \wedge q_{a_i} = q_{a_i}$ for all $a_i \in A$, $i \in \{1, 2, 3, 4\}$. Also, from Table~\ref{table:d-example}, $C(q_{a_i}) \neq \emptyset$, for all $a_i \in A$. Let $A'$ be any of the 14 sets in the table above, and let $A'' \neq A'$ be any of the remaining 13 sets. Then, according to Lemma~\ref{lem:twopart}, $C(q_{A'}) \neq C(q_{A''})$. One can easily verify through Table~\ref{table:d-example} that this is indeed true.\qed 
\end{example}

Let $n = |C(q_\mathbf{b})| = |C(q_\mathbf{b} \wedge q_A)|$, which we seek to find through the attack. Consider a partition $\{A_1, A_2\}$ in $P_A$, and note that $(q_\mathbf{b} \wedge q_{A_1})(D) + (q_\mathbf{b} \wedge q_{A_2})(D) = n$.
Now consider the queries $(\mathbf{b}, A_1)$ and $(\mathbf{b}, A_2)$ to the Attribute Analyser. In return, among other answers, we get $\alpha_{A_1}$ and $\alpha_{A_2}$, which are the noisy answers to the two (total) queries mentioned above. Adding the two, we have $z = n + e_1 + e_2$, where $e_1$ and $e_2$ are unknown error terms from $\mathbb{Z}_{\pm r}$. Our attack is as follows: for each of the $k = 2^{m - 1} - 1$ partitions in $P_A$, query the Attribute Analyser with the two sets in the partition, add the answers, and average them over all $k$. The algorithm is shown in Algorithm~\ref{algo:noise-remover}.

\begin{algorithm}[!ht]
\SetAlgoLined
\SetAlCapSkip{1em}
\DontPrintSemicolon{}
\let\oldnl\nl
\newcommand{\nonl}{\renewcommand{\nl}{\let\nl\oldnl}}
\SetKwInOut{Input}{Input}
\SetKwInOut{State}{State}
\Input{A vector of attribute values $\mathbf{b} \in \mathbf{B}$, attribute $A$ with $m$ different attribute values and set $P_A$ of two-partitions of $A$.}
Initualize $z \leftarrow 0$.\;
\For{each two-partition $\{A_1, A_2\}$ in $P_A$}{
	Query the Attribute Analyzer with inputs $(\mathbf{b}, A_1)$ and $(\mathbf{b}, A_2)$ and obtain $\alpha_{A_1}$ and $\alpha_{A_2}$.\;
	Update $z \leftarrow z + \alpha_{A_1} + \alpha_{A_2}$.\;
}
Let $k = 2^{m - 1} - 1$ and obtain $z \leftarrow z/k$.\;
Output $\lfloor z \rceil$.\;
\caption{The Noise Remover Algorithm.}
\label{algo:noise-remover}
\end{algorithm}

Notice that in each loop the Attribute Analyser queries Bounded Noisy Counts $\mathcal{M}$ twice. Therefore there are a total of $t = 2k = 2^m - 2$ queries to $\mathcal{M}$.

\subsubsection{Success Probability}
\label{subsub:succ-prob}
Let $Z_i$ denote the random variable denoting the sum in Step 4 of the algorithm for the $i$th partition, where $i  \in [k]$, $k = 2^{m - 1} - 1$. We have
\begin{equation}
\label{eq:zi}
Z_i = n + E_1^{(i)} + E_2^{(i)},
\end{equation}
where $E_1^{(i)}$ and $E_2^{(i)}$ are the noise variables. 
\begin{lemma}
\label{lem:twopart:ind}
For each $i \in [k]$, $E_1^{(i)}$ and $E_2^{(i)}$ are i.i.d. random variables with distribution $\mathbb{U}_{\pm r}$. Furthermore, $Z_1, \ldots, Z_k$ as defined by Eq.~\ref{eq:zi} are i.i.d. random variables.
\end{lemma}
\begin{proof}
See Appendix~\ref{app:proofs}.
\end{proof}
Now define $\overline{Z} = \frac{1}{k} \sum_{i = 1}^k Z_i$. The success probability of the Noise Remover is then given by $\Pr \left( \left\lvert \overline{Z} - n \right\rvert < 0.5 \right)$. Define $Y_i =  E_1^{(i)} + E_2^{(i)}$ and $\overline{Y} =  \frac{1}{k} \sum_{i = 1}^k Y_i $.Then
\begin{align}
\Pr \left( \left\lvert \overline{Z} - n \right\rvert < 0.5 \right) &= \Pr \left( \left\lvert \overline{Y} + \frac{1}{k} \sum_{i = 1}^k n - n \right\rvert < 0.5 \right) \nonumber\\
																						&= \Pr \left( \left\lvert \overline{Y} \right\rvert < 0.5 \right)  \nonumber
\end{align}
Thus, we will attempt to find $\Pr \left( \left\lvert \overline{Y} \right\rvert < 0.5 \right)$. We will first show a lower bound on this probability and then an exact expression.

\descr{Lower Bound on the Success Probability.} Using Chebyshev's inequality, we see that
\begin{equation}
\label{eq:cheby}
\Pr \left( \left\lvert \overline{Y} - \mathbb{E}(\overline{Y}) \right\rvert \ge \epsilon \right) \le \frac{\text{Var}(\overline{Y})}{\epsilon^2}.
\end{equation}
By setting $\epsilon = 0.5$, and putting in $k = 2^{m-1} - 1$, and the values of $\mathbb{E}(\overline{Y})$ and $\text{Var}(\overline{Y})$ (see Appendix~\ref{appsub:lower-bound}), we get
\begin{equation}
\label{eq:cheby-uniform}
\Pr \left( \left\lvert \overline{Y}  \right\rvert < 0.5 \right) \ge 1 - \frac{8r(r+1)}{3(2^{m-1} - 1)}.    
\end{equation}

Figure~\ref{fig:bounded-nr-succ} shows lower bounds on the success probabilities against different perturbation parameters as a function of $t = 2k$, i.e., the number of queries to Bounded Noisy Counts. 

\begin{figure}[!ht]
\centering
\ifpets
\else
\includegraphics[scale=0.8]{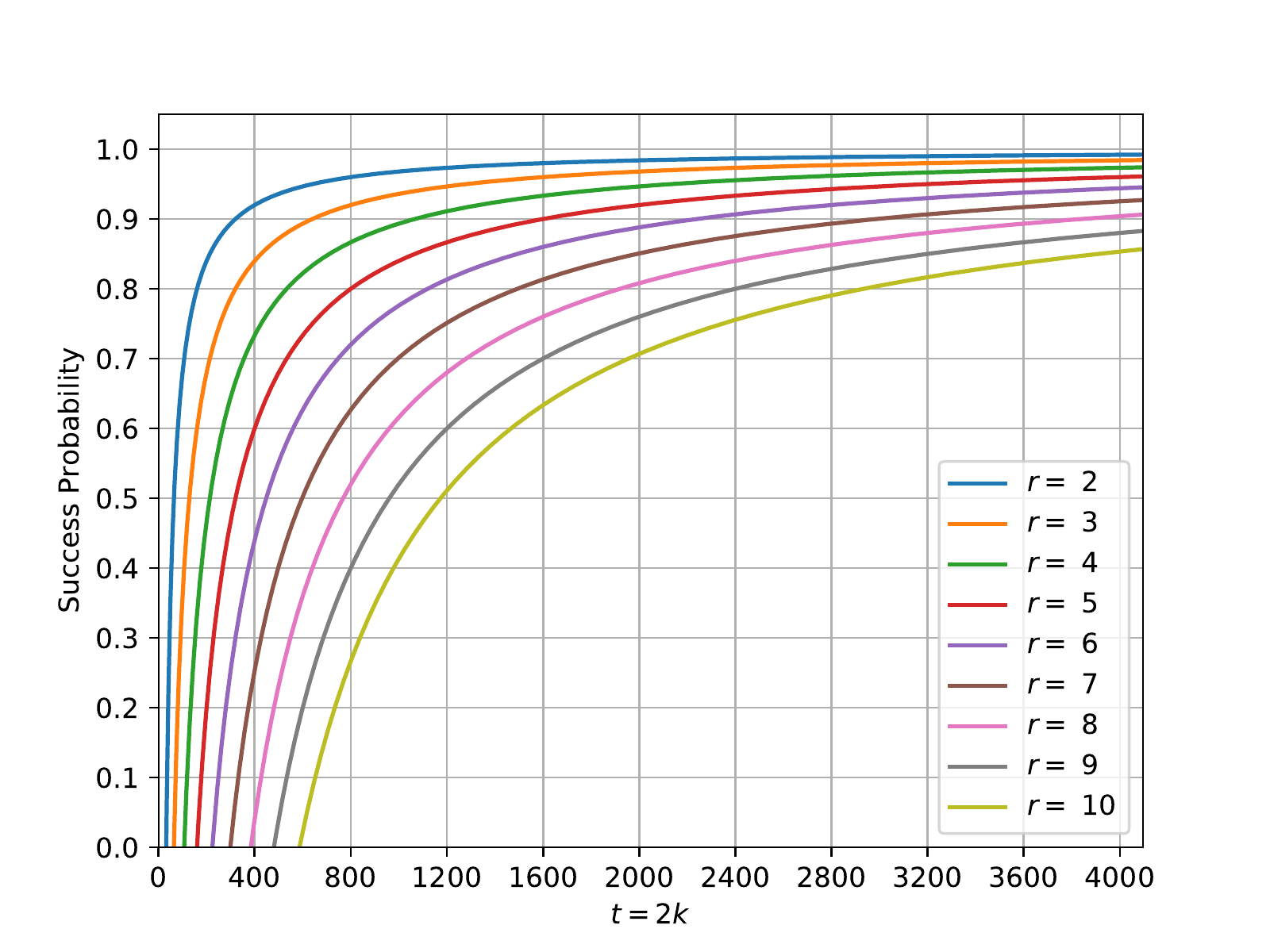}
\fi
\caption{Lower bounds on the probability of successfully retrieving the actual count $n = (q_{\mathbf{b}} \wedge q_A)(D)$ through the Noise Remover algorithm. Here $m = 12$ and hence $t = 2k$ ranges from $2$ to $2(2^{m - 1} - 1) = 4094$.}
\label{fig:bounded-nr-succ}
\end{figure}

\begin{remark}
Figure~\ref{fig:bounded-nr-succ} shows that we do not need to use all the two-partitions in $P_A$ to achieve a given probability of success. Furthermore, since each iteration calls the Attribute Analyser twice (one for each partition), we have a total of $2k$ calls to Attribute Analyser. Thus, if we were to run this algorithm on the TBE algorithm via the TableBuilder tool, this would require running the tool a total of $2k$ times. 
\end{remark}

\descr{Exact Success Probability.} Consider the sum $\sum_{i = 1}^k E_1^{(i)} + E_2^{(i)}$. Simplifying notation, we can view this as the sum of $2k$ i.i.d. random variables $E_i$ (due to Lemma~\ref{lem:twopart:ind}). The probability mass function of each $E_i$ is given by
\[
f_{E}(x) = 
\begin{cases}
\frac{1}{2r+1} &  \text{ if } x \in \mathbb{Z}_{\pm r}\\
0  & \text{ otherwise}
\end{cases}.
\]
In Appendix~\ref{appsub:exact}, we show that
\begin{equation}
\label{eq:nr:suc:exact}
\Pr \left( \left\lvert \overline{Y}  \right\rvert < 0.5 \right) = \sum_{x \in (-k/2, k/2)} f_{X_{2k}} (x),
\end{equation}
where $X_{2k}$ is the sum of $2k$ i.i.d. random variables $E_i$. Thus, we can evaluate Eq.~\ref{eq:nr:suc:exact} to find the exact success probability of the Noise Remover algorithm to obtain the answer $n = (q_{\mathbf{b}} \wedge q_A)(D)$. Figure~\ref{fig:sim-nr-succ} shows these success probabilities. Comparing this with Figure~\ref{fig:bounded-nr-succ}, we see that the actual success probability is higher for much smaller values of $t = 2k$.
\begin{figure}[ht!]
\centering
\ifpets
\else 
\includegraphics[scale=0.8]{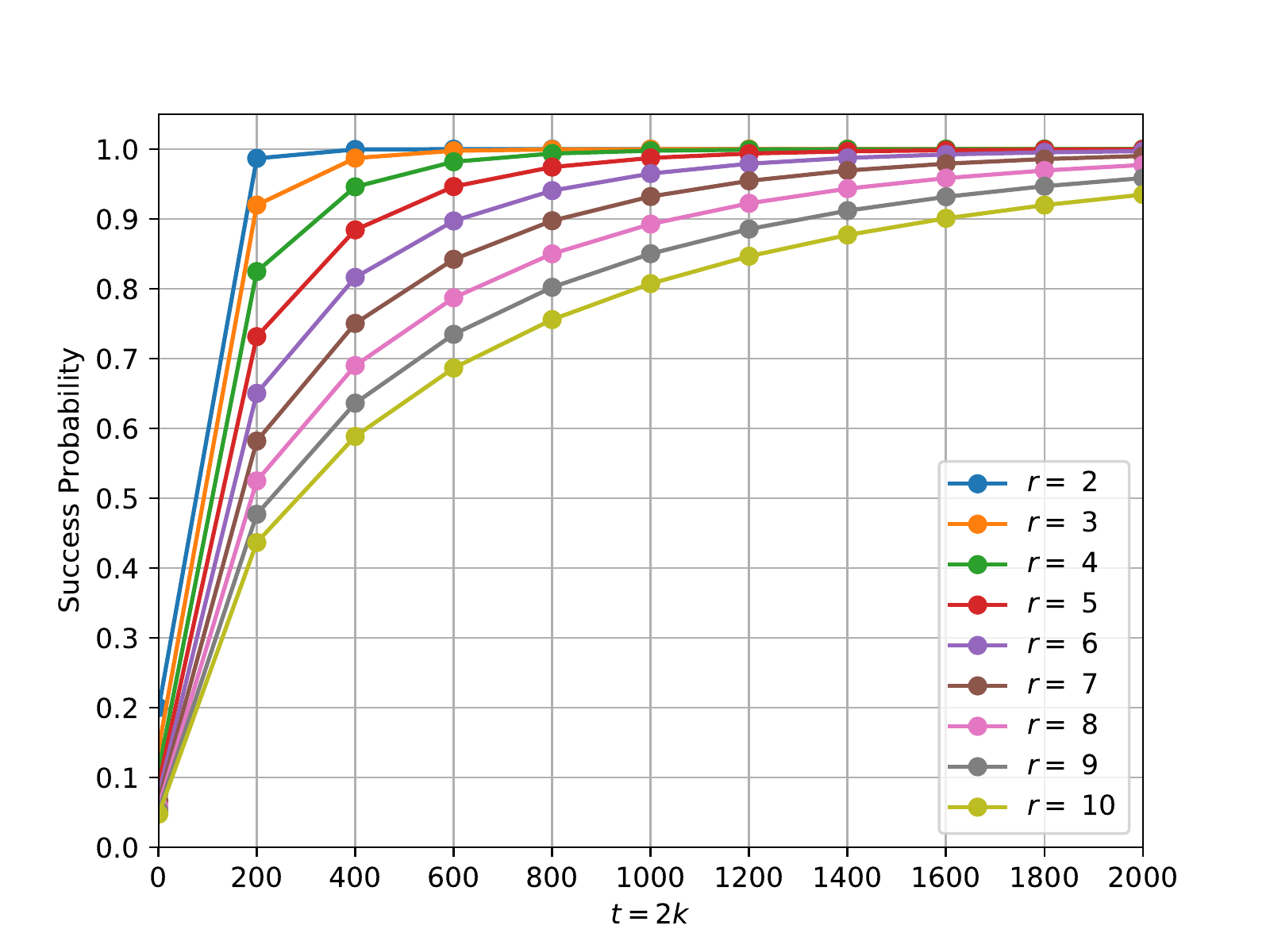}
\fi 
\caption{Probability of successfully retrieving the actual count $n = (q_{\mathbf{b}} \wedge q_A)(D)$ through the Noise Remover algorithm. Here $t = 2k$ ranges from $2$ to $2000$.}
\label{fig:sim-nr-succ}
\end{figure}

\subsubsection{Broadening the Scope of Attack 2}
\label{subsub:broaden}
We now show that our attack can be used much more broadly.

\descr{Relaxing the Non-Zero Outputs Assumption.} We assumed in the previous section that all $m$ attribute values of $A$ satisfy $\mathcal{M}(q_{\mathbf{b}} \wedge q_{a_i}) \neq 0$, $i \in [m]$. First note from Figure~\ref{fig:sim-nr-succ} that we do not need all the two-partitions of $A$ to find $q_{\mathbf{b}} \wedge q_A$. The only requirement is to have a sufficient number of two-partitions $t$ to ``average out'' $n$. Thus, as long as we have $m' \le m$ number of attribute values whose corresponding queries have non-zero answers (via Bounded Noisy Counts), 
we can use them to find the answer to the aforementioned query in the following way. Let $a_1, \ldots, a_{m'}, a_{m' + 1}, \ldots, a_m$ denote the $m$ attribute values of $A$, and assume (w.l.o.g.) that only upto $a_{m'}$ have $\mathcal{M}(q_{\mathbf{b}} \wedge q_{a_i}) > 0$. 
Therefore $\mathcal{M}(q_{\mathbf{b}} \wedge q_{a_i})$ is $0$ for all $m' < i  \le m$. Let $A''$ denote the set of attribute values $a_{m' + 1}, \ldots, a_m$. Note that $C(q_{\mathbf{b}} \wedge q_{A''})$ can be possibly empty. We first construct all two-partitions of the set $A' = A - A''$, resulting in $2^{m' - 1} - 1$ two-partitions. Denote this by $P_{A'}$. Then in each two-partition we add $A''$ to any one (but not both) of the two sets in the partition. It is easy to see that the resulting set is a set of two-partitions of $A$ (not necessarily the set of \emph{all} two-partitions of $A$). Furthermore, we still ensure that Lemma~\ref{lem:twopart:ind} holds. For, if $C(q_{\mathbf{b}} \wedge q_{A''})$ is empty, then Lemma~\ref{lem:twopart:ind} automatically holds due to construction of $P_{A'}$. On the other hand, if $C(q_{\mathbf{b}} \wedge q_{A''})$ is not empty, then we are adding a set of new contributors which are not in $A'$. Adding these contributors means that Bounded Noisy Counts adds fresh noise to the corresponding query (on the particular set in the given two-partition in which $A''$ is added). Thus, Lemma~\ref{lem:twopart:ind} follows due to Lemma~\ref{lem:twopart} in this case. We shall call $A'$ a subset of $A$ with \emph{non-zero answers}.

\descr{Removing the Noise on a Target Attribute Value.} Let us now assume that we are interested to know the value $q_\mathbf{b} \wedge q_{a}$ for some target attribute value $a$ in $A$. We take a subset $A'$ of $A$ with non-zero answers such that $a \notin A'$. We first 
run the Noise Remover on the set of two-partitions $P_{A'}$ of $A'$, obtaining count $n'$. We then construct $P_{A' \cup \{a\} }$, and run the Noise Remover algorithm again to obtain the count as $n''$. The 
answer to the above query is then $n'' - n'$. If  $\mathcal{M}(q_\mathbf{b} \wedge q_{a}) = 0$ then we can use the trick mentioned above to construct two-partitions of $A' \cup \{a \}$. Let $p_{\text{nr}}$ denote the probability of success of Noise Remover. Then, through a simple application of the union bound, the success probability is given by $1 - 2(1 - p_{\text{nr}})$. This requires around $2(2k) = 4k$ calls to the Attribute Analyser, and hence $t = 4k$ calls to Bounded Noisy Counts.

\descr{Removing the Noise on the Attribute Histogram.} Continuing on with the previous example we can in fact find answers to $q_\mathbf{b} \wedge q_{a_i}$ corresponding to all $m$ attributes of $A$. We first construct a subset $A'$ of $A$ with non-zero answers (with $m'$ number of attributes), and use Noise Remover once to find $q_\mathbf{b} \wedge q_{A'}$. For any attribute $a \in A - A'$ we follow the methodology defined above to retrieve the answer. Thus, a further $m - m'$ calls to the Noise Remover. On the other hand, for any target attribute value $a' \in A'$, we run the Noise Remover on two-partitions of $A' - \{a' \}$ (instead of $A' \cup \{a'\} = A'$). This means a further $m'$ calls to the Noise Remover. By the union bound, the overall probability of success is given by $1 - (m - m' + 1 + m')(1 - p_{\text{nr}}) = 1 - (m + 1) (1 - p_{\text{nr}})$. This requires $2k(m + 1)$ calls to the Attribute Analyser, and hence $2k(m+1)$ queries to Bounded Noisy Counts. Figure~\ref{fig:union-nr-succ} shows the success probability in finding all queries corresponding to all attribute values in some target attribute $A$. Here we have used $k = 800$ ($t = 1600$), and thus $|A'|$ has to be $\ge 12$. Recall that $k$ cannot be greater than $2^{m' - 1} - 1$. While $m' = 11$ would suffice to find all attribute values in $A - A'$, we require $m' = 12$ so that we can find the attribute values within $A'$ as well. Note that this result is obtained through the union bound, and the actual success probability is likely to be much better. 
\begin{figure}[ht!]
\centering
\ifpets
\else
\includegraphics[scale=0.6]{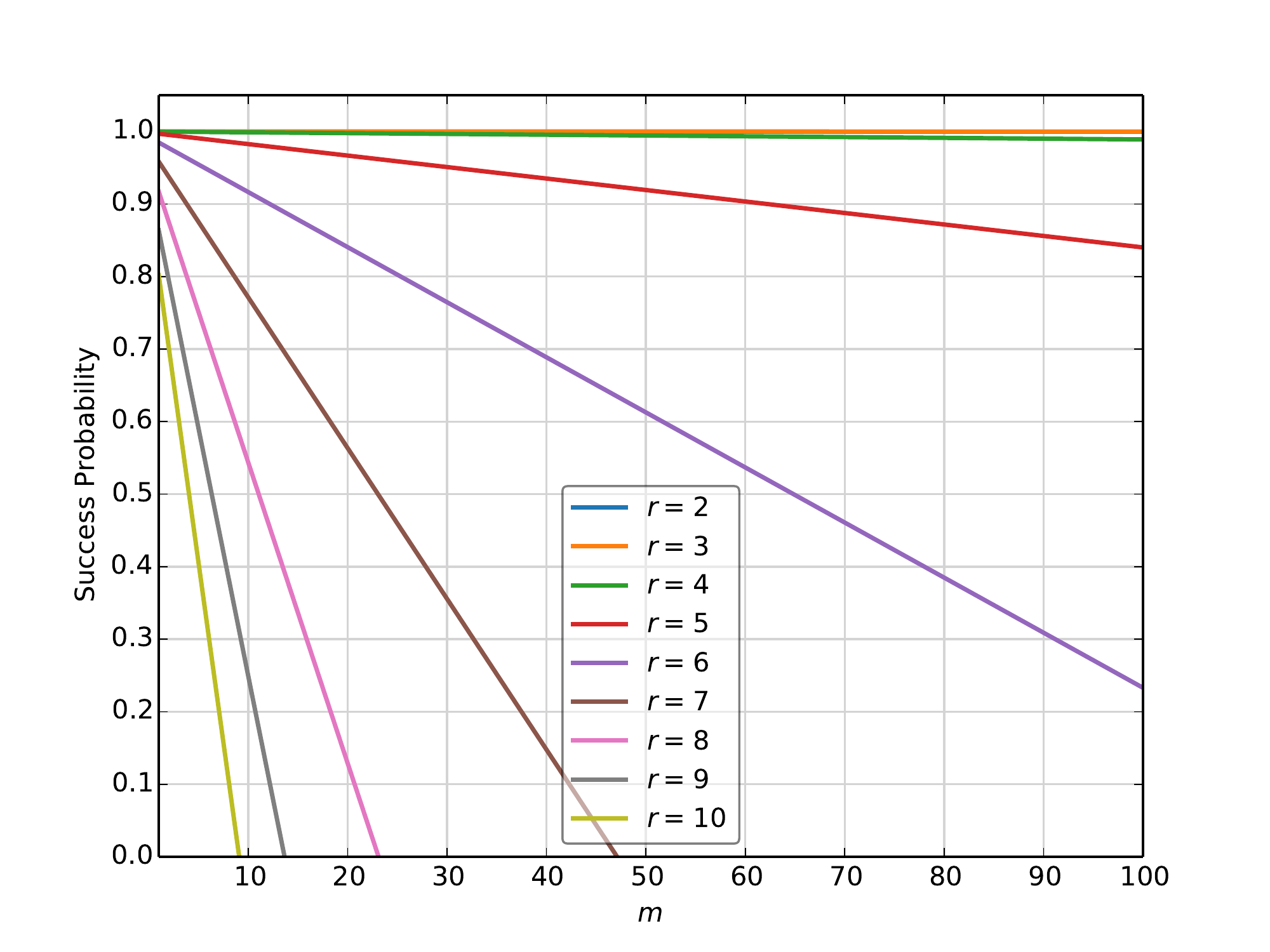}
\fi
\caption{Probability of successfully retrieving all actual counts of the queries $(q_{\mathbf{b}} \wedge q_{a_i})(D)$ of an attribute $A$ with $m$ different attribute values. Here $k = 800$ and hence $t = 1600$.}
\label{fig:union-nr-succ}
\end{figure}

\section{Experimental Evaluation on the TBE Algorithm}
\label{sec:eval}

\begin{table*}[!ht]
\centering
\caption{Results of running Attack 2 to retrieve a target column of the synthetic dataset. Cells labeled * means that some instances had negative counts which were floored to 0. There were four such instances for $k = 50$, and one instance for $k = 127$. In all cases the returned count was $-1$.}
\label{tab:results}
\begin{tabular}{c|c|c|c|c|c}
\multirow{2}{*}{Count ($c$)} & True & \multicolumn{4}{c}{Total Correctly Retrieved}\\
\cline{3-6}
& Instances & $k = 50$ & $k = 127$ & $k = 200$ & $k = 255$ \\
\hline\hline
$c = 0$ & $56$ & $55^*$ & $56^*$ & $56$ & $56$ \\
$0<c\le 4$ (suppressed) & $6$ & $6$ & $6$ & $6$ & $6$\\
$4 < c \le 100$ & $8$ & $8$ & $7$ & $8$ & $8$\\
$100 < c \le 44865$ & $37$ & $34$ & $37$ & $37$ & $37$\\
\hline
Total & $107$ & $104$ & $106$ & $107$ & $107$\\
Success Percentage & - & $97.2\%$ & $99.1\%$ & $100\%$ & $100\%$\\
\end{tabular}
\end{table*}

\subsection{Synthetic Dataset}
We ran Attack 2 on a synthetic dataset accessed via an API built on top of the TBE algorithm. The API mimics the functionality of the TableBuilder tool from ABS. Our privacy algorithms represent an abstract mathematical model of the TBE algorithm. As such there is one significant simplification used in our mathematical model that needs specific mention. Since the TBE algorithm is meant to answer queries \emph{on-the-fly}, it maintains a pre-computed table of noise (instead of freshly generated noise for queries with a new set of contributors). Since the number of queries can be much larger than the dimension of the table, a mechanism is introduced that deterministically accesses the relevant noise entry in the table. The entries in the table are themselves derived from an \emph{admissible} distribution that maximises entropy subject to utility constraints~\cite{marley-tb}. In the next section (Section~\ref{sec:entropy}), we show that the discrete uniform distribution over $\mathbb{Z}_{\pm r}$ is one such distribution, and any other choice of admissible distribution is still susceptible to our attack. Thus, the cells in the table of noise can be safely assumed to be uniform random entries from $\mathbb{Z}_{\pm r}$. To ensure that same contributors receive the same noise, the contributors (users) in the dataset are assigned unique keys. When combining different contributors, the keys are XORed and then given as input to a pseudo-random number generator which in turn maps it to a perturbed value in the table~\cite{tb-differencing, marley-tb}. We see that with a big enough perturbation table, our model is a good approximation. As we shall see, the results of our attack confirms this. The aforementioned API uses the perturbation parameter $r = 2$. Furthermore, it returns the output ``suppressed'' for counts of $\le 4$. However, for the sake of our attack, we assume that the returned count is 0 (which is a more difficult problem). 

We fix a target attribute in the synthetic dataset. The attribute has $107$ different attribute values. Details of the synthetic dataset including how it was generated and the resulting counts are given in Appendix~\ref{app:syndata}. We then run the Noise Remover on each attribute value with different values of $k$ (number 
of two-partitions). The results are shown in Table~\ref{tab:results}. As an example, with $k = 200$, the probability that all $107$ attribute values are returned correctly is at least $0.959$ (according to the analysis above, using a union bound). In practice 
with $k = 200$ and $k = 255$ all attribute values are returned exactly without any error. We see that even with $k = 50$ (which amounts to $t = 2k = 100$ queries to Bounded Noisy Counts per attribute value), we have only a 7.48\% of 
attribute values with an incorrect answer. 

A few observations are in order: (a) First, even for the cases where the actual count returned is incorrect, the level of noise is reduced (i.e., it is $\pm 1$ 
instead of $\pm 2$). (b) Secondly, in some cases the noise returned is $-1$. By the properties of the algorithm, we can fix this to $0$. This results in even less percentage of erroneous attribute values: 3.74\% 
error for $k = 50$ and $0.9\%$ (only one incorrect guess) for $k = 127$. (c) A final observation is that the probabilities reported in Figure~\ref{fig:union-nr-succ} are for the entire column. If the target is only one specific set of attribute values (corresponding to a target individual), then the probabilities are higher, i.e., $0.9996$ for the case of $k = 200$.

In Appendix~\ref{app:constraints}, we experimentally analyse the success probability of the Noise Remover by varying the number of attribute values whose corresponding queries have non-zero outputs, and discuss some workarounds when this number is low. 

\subsection{Adult Dataset}
We also ran the attack on a real-world dataset. For this, we used the Adult dataset~\cite{adult}, which is an extract from the 1994 US Census information.\footnote{More specifically, we use the \texttt{adult.data} file from \url{https://archive.ics.uci.edu/ml/datasets/Adult}.} The dataset consists of 32,561 rows, each containing an individual's data. We extract the age column for the attack. This column contains all ages in the integer range $[17..88]$ and the age 90. We augment this by including ages 10 to 16 inclusive, age 89, and ages 91 to 120 inclusive, each obviously having a count of 0. Thus, there are a total of 111 values for the age attribute to be queried via the TBE API. The breakdown of true counts is as follows: 

\begin{center}
\begin{tabular}{c|c}
{Count ($c$)} & True Instances\\
\hline\hline
$c = 0$ & $38$ \\
$0<c\le 4$ (suppressed) & $4$ \\
$4 < c \le 100$ & $16$\\
$100 < c \le 898$ & $53$\\
\hline
Total & $111$ \\
\end{tabular}
\end{center}

We used the set of 10 attribute values $A' = \{17, 18, \ldots, 27\}$ as the \emph{base set}. All of the attribute in this set have a true count of $\ge 395$, and hence their noisy counts would not be suppressed by Bounded Noisy Counts for any reasonable perturbation parameter value. We then used 1,000 two-partitions out of the possible $1,023$ two-partitions from $A'$ to find $q_{A'}$ through the Noise Remover. We denote this query answer by $n'$. Then for each age $a \notin A'$ we create the set $A'' = A' \cup \{a\}$, and use the Noise Remover with $k$ two-partitions to find the answer to $q_{A''}$ as $n''$. The answer to $q_a$ is then obtained as $n'' - n'$. For an age $a \in A'$, we create the set $A'' = A' - \{a\}$, and again use $k$ two-partitions to find the answer to $q_{A''}$ via Noise Remover as $n''$. The answer to $q_a$ in this case is $n' - n''$. 

Notice that only for the base set do we use a total of 1,000 two-partitions, since this will be done only once. The number $k$ of two-partitions used to compute the answers to $q_{A''}$ for each $A''$ is from the set $\{50, 100, 200, 250\}$. For each $k$, we run the experiment a total of 100 times. The average success rate is then reported. For this experiment we use three different perturbation parameters: $r= \{2, 3, 5\}$. The results are summarised in Table~\ref{tab:adult-results}. We see that even with $r = 5$, the attack successfully recovers the true count of more than $93\%$ of attribute values with $k = 250$ two-partitions used per attribute value, i.e., $t = 2k = 500$ queries per attribute value to Bounded Noisy Counts.

\begin{table*}[!ht]
\centering
\caption{Results of running Attack 2 to retrieve the age column of the Adult dataset against different perturbation parameter values. Negative counts were ceiled to 0. For each $k \in \{50, 100, 200, 250\}$ (number of two-partitions) the average of 100 runs is reported.}
\label{tab:adult-results}
\begin{tabular}{c|c|c|c|c|c}
Perturbation & \multirow{2}{*}{Total} & \multicolumn{4}{c}{Total Correctly Retrieved}\\
\cline{3-6}
$r$ & & $k = 50$ & $k = 100$ & $k = 200$ & $k = 250$\\
\hline\hline

$\pm 2$ & $111$ & $103.2$ & $110.1$ & $111.0$ & $111.0$ \\
& - & $93.0\%$ & $99.2\%$ & $100.0\%$ & $100.0\%$ \\
\hline
$\pm 3$ & $111$ & $89.8$ & $103.9$ & $110.0$ & $110.8$ \\
& - & $80.9\%$ & $93.6\%$ & $99.1\%$ & $99.8\%$ \\
\hline
$\pm 5$ & $111$ & $70.2$ & $88.0$ & $98.2$ & $103.6$ \\
& - & $63.3\%$ & $79.3\%$ & $88.4\%$ & $93.4\%$ \\
\end{tabular}
\end{table*}

\section{Some Inherent Limitations}
\label{sec:entropy}
The noise distribution for the TBE algorithm is required to maximise disclosure control subject to utility constraints~\cite{marley-tb}. We call such a distribution, an \emph{admissible distribution}. More precisely, given a finite set of integer perturbation values $\Pi$, an admissible distribution $\mathcal{E}$ can be obtained by maximising the entropy $
- \sum_{e \in \Pi} p(e) \log_2 p(e)$ subject to the constraints
\begin{enumerate}
    \item It should be a probability distribution: $\sum_{e \in \Pi} p(e) = 1$.
    \item It should be unbiased, i.e., $\mathbb{E}(E) = 0$, where $E$ is the random variable distributed as $\mathcal{E}$.
    \item It should have bounded variance, i.e., $\text{Var}(E) \le v$, for some threshold $v$.
\end{enumerate}
These properties are stated in~\cite{marley-tb}, except that we have excluded the condition that the noise values should not be less than $0$ or a positive value, as we expect this to be handled by the suppression parameter. The method of Lagrange multipliers~\cite[p. 707]{bishop-pattern} can be used to solve the above problem to find the distribution $\mathcal{E}$ given the set $\Pi$~\cite{marley-tb}.

We call $\mathcal{E}$ non-trivial if it is any distribution other than $p(0) = 1$. An immediate consequence of constraint 2 is that any non-trivial distribution requires the set $\Omega$ to have at least one negative and one positive integer. Also, if we let $v \ge r(r+1)/3$ in constraint 3, then given the set $\Pi = \mathbb{Z}_{\pm r}$, we get the discrete uniform distribution $\mathbb{U}_{\pm r}$ as the unique solution to the above optimisation problem (recall that $r(r+1)/3$ is the variance of the discrete uniform distribution $\mathbb{U}_{\pm r}$). This is the distribution that we have used in this paper. Other admissible distributions include, the zero-mean truncated normal distribution, the zero-mean truncated Laplace distirbution, and in general any truncated zero-mean symmetric distribution. Samples from these distributions can be rounded to nearest integers to fall in the set $\Pi$.

A natural question to ask is whether any other admissible distribution makes the attack significantly harder. The answer to this question is negative. The clue lies in Eq.~\ref{eq:cheby}. First, if we let $v < r(r+1)/3$, then the resulting distribution will have variance less than $\mathbb{U}_{\pm r}$, and through Eq.~\ref{eq:cheby}, the variance of the average of $t = 2k$ such variables will be less than its uniform counterpart, and as a result the attack requires fewer queries (given by $t$) to remove noise. We can in fact derive the minimum condition required of the perturbation to withstand the attack. First, from Section~\ref{subsub:succ-prob}, we see that $    \text{Var}(\overline{Y}) = \frac{1}{k^2} \sum_{i = 1}^k \text{Var}(Y_i) = \frac{1}{k}\text{Var}(Y)$, where $Y$ is the random variable denoting the sum of two random variables distributed as $\mathcal{E}$, and $k$ is the number of two-partitions. The two being i.i.d., we get $\text{Var}(\overline{Y}) = \frac{2}{k}\text{Var}(E)$, where $E$ is distributed as $\mathcal{E}$. Since $\mathbb{E}(E) = 0$, we have that 
\[
\text{Var}(E) = \sum_{e \in \Pi} e^2 p(e) \le \sum_{e \in \Pi} c^2 p(e) = c^2,
\]
where $c = \max\{ |e| : e \in \Pi \}$, i.e., the maximum absolute perturbation. Thus, we get $\text{Var}(\overline{Y}) \le  \frac{2c^2}{k} = \frac{4c^2}{t}$, where $t = 2k$ is the number of queries to the Bounded Noisy Counts algorithm. Now, Eq.~\ref{eq:cheby} shows that for the attack to be unsuccessful, we should have $\text{Var}(\overline{Y}) = \Omega (1)$. Together with the previous result, this implies that, we require
\begin{equation}
    c = \Omega(\sqrt{t}).
\end{equation}
Thus, the amount of perturbation needs to be of order $\sqrt{t}$ to thwart the attack, where $t$ equals the number of queries. This result is consistent with the results from linear reconstruction attacks~\cite{dinur-nissim}. Incidentally, this is also the level of noise required by any differentially private algorithm to answer $t$ queries (without coordinated answers to queries)~\cite{power-of-state, vadhan-tutorial}. 
For instance, if the privacy budget $\epsilon$ is a small constant, then adding zero-mean Gaussian noise of standard deviation $\sqrt{t}$ to each of the $t$ queries satisfies \emph{concentrated differential privacy}~\cite{conc-dp}. Obviously, the noise is of scale $O(\sqrt{t})$. Likewise, to achieve the notion of $(\epsilon, \delta)$-differential privacy, with a constant $\epsilon$ and negligible $\delta$, one can answer $O(t)$ arbitrary counting queries with noise of scale $O(\sqrt{t})$ via the Laplace mechanism using the advanced composition theorem~\cite{steinke-pure-approx},\cite[Theorem 7.2.7]{vadhan-tutorial}. A consequence of the above result is that any perturbation algorithm with bounded (constant) noise is eventually expected to succumb to our noise removing attack.  
\section{Mitigation Measures}
We briefly discuss possible mitigation measures against the two attacks separately. 
\subsection{Mitigation Measures against Perturbation Finder}
Recall that the attack algorithm on finding the perturbation value relies on identifying an attribute with at least two attribute values. Assume this to be the attribute gender, with attribute values male and female. The attack involves submitting queries on the number of males, the number of females, \emph{and} the combined number of males and females (total query). 

\descr{Query Auditing.} A first defence mechanism is to audit queries to check if an analyst is attempting to find the perturbation parameter. This measure needs to identify all possible query combinations that can be used to narrow down the possible range (of the perturbation parameter). The specific construction of queries (outlined above) in our attack is one possible way. However, there may exist other combination of queries which could be used to find the perturbation parameter. This requires an exhaustive analysis. Furthermore, it is difficult to detect if there is malicious intent behind a given series of queries, as they can be contextually benign, e.g., an analyst might very well be checking the gender distribution across different occupations in a geographic area. In general, query auditing is a difficult problem~\cite{dp-book}.  

\descr{Query Throttling.} Another alternative is to throttle the number of queries. This can be done by introducing a ``cap'' on the number of queries allowed to an analyst. However, in light of our results, this would be too small a number, e.g., not allowing more than 200 queries if $r = 5$ is used as the perturbation parameter. 

\descr{Eliminating the Total Query.} Recall that the attack algorithm works by examining the difference between the noisy count of the total query versus the sum of noisy counts of the sub-queries. Thus one way to mitigate the attack is to \emph{not} add ``fresh'' noise to the total query (and instead report the sum of the noisy counts from the sub-queries). Unfortunately, this significantly impacts utility. For instance, if an analyst is interested in the number of people living in a certain geographic area (say the suburb Redfern), then the only way to obtain this answer would be to add the answers obtained from the number of males and the number of females living in the area. The problem is further exacerbated by the fact that there might be multiple attributes with two attribute values under the same geographic area. And thus the attack can be (slightly) modified to instead equate the sums obtained from multiple pairs of attribute values. 

\descr{Disclosing the Perturbation Parameter.} In light of the shortcomings of the above mentioned defence measures, an inevitable choice is to make the perturbation parameter public. Apart from having a negligible impact on (individual) privacy, this is beneficial from a utility point of view as well. The analyst now knows the degree to which an answer is possibly perturbed, and can factor this amount into his/her calculations.

\subsection{Mitigation Measures against Noise Remover}
Recall that the attack on removing noise relies on creating two-partitions of a target attribute, and the fact that fresh noise is added to the answers to the total queries from the two-partitions.  
  
\descr{Query Auditing.} Automated checks could be applied to see if a significant number of queries correspond to different two-partitions of the same sub-population. Several issues make this a less than ideal solution. First, malicious queries might not be successively submitted; a clever attacker might inject these queries in between several innocuous queries. In general, query auditing can be computationally infeasible; indeed, it is NP-Hard to detect maliciously crafted queries via query auditing~\cite{query-audit}. In fact, the need to dispense with query auditing is one of the motivations behind the rigorous definition of differential privacy~\cite{dp-book}. Secondly, we could modify the attack to include three-partitions instead of two-partitions (with a corresponding increase in the number of queries required to remove noise). Lastly, while we have demonstrated one way in which multiple answers can be combined together and averaged to remove noise, we have not checked and confirmed whether there exist other query combinations which could do the same.  

\descr{Query Throttling.} Placing a cap on the allowed number of queries is another option, with the obvious drawback that it limits the analyst to a much smaller number of queries. The attacks described in this paper as well as prior work on reconstruction attacks~\cite{dinur-nissim} suggest that this is unavoidable if bounded noise mechanisms are deemed indispensable. A technical difficulty is proposing a quantitative bound on the number of queries allowed. For instance, our results show that even 100 queries remove the noise for most attributes with a perturbation parameter of $r = 2$. 

\descr{Eliminating the Total Query.} The noise removing algorithm relies on the fact that answer to the total query adds fresh noise, which can be compared against the sum of the noisy counts of queries corresponding to the two-partitions. If the total query does not add fresh noise, the sum would be noisier and as a result would require a larger number of queries to eliminate noise. However, as discussed before, this is not desirable from a utility point of view. For instance, if the analyst wishes to know the number of people in a sub-population with age greater than 50, then the only way to obtain this would be to add the result obtained from each age grouping (and thus obtain a noisier answer). We do note that there are techniques to construct differentially private histograms which maintain \emph{consistency}~\cite{barak-fourier, hist-consistency}, with often the output being much more accurate~\cite{hist-consistency}. Here consistency, for instance, means that the noisy answer to the total query and the total after summing the noisy counts of per-attribute value queries are the same. These approaches rely on optimization via post-processing. For instance, one way to maintain consistency is to add noise to queries on attribute value tuples from all the attributes in the dataset, and answer queries on fewer attributes from these~\cite[\S 3.1]{barak-fourier}. Thus, there are better ways to answer the total query than the one outlined above.

\descr{Provably-Private Alternatives.} Our attack is another example of a series of attacks demonstrating that extremely high accuracy cannot be guaranteed for too many queries due to privacy concerns. For instance, prior results have shown that noise needs to be calibrated according to the number of queries to avoid database reconstruction attacks~\cite{exposed, dinur-nissim}. In Section~\ref{sec:entropy}, we showed that this is needed to prevent our attacks as well. Thus, a safe way of releasing noisy answers is to scale the noise as a function of the number of queries asked. Differential privacy~\cite{calib-noise, dp-book} is a privacy definition and framework that allows to do this. The parameter $\epsilon$ in differential privacy determines the noise added to query answers and can be tweaked to find a balance between privacy and utility. Furthermore, this parameter can be safely disclosed without effecting privacy. However, this suggests that answering too many queries will result in noise that badly affects utility. This is an inherent limitation of \emph{any} privacy-preserving mechanism. In particular, there is growing amount of evidence (including this work) that suggests that any meaningful guarantee of privacy cannot allow extremely accurate answers to an unlimited number of queries. 

\subsection{Lessons Learned} 
We summarise some of our recommendations:
\begin{itemize}
    \item Parameters used in the privacy algorithms, e.g., the perturbation parameter, should not be kept secret. It is easy to retrieve them if hidden, and, furthermore, disclosing them upfront is helpful for analysts.
    \item Ad hoc workarounds to mitigate averaging attacks, e.g., by adding same noise to same queries or same contributors, only marginally impede them. At best they inform us where fresh noise should not be wasted, e.g., if same query is repeated.
    \item Mitigating attacks with specific patches, e.g., eliminating the total query, may not prevent other ways of carrying out these attacks.
    \item Adding fresh noise to query answers is not a privacy problem as long as the noise scale is a function of the number of queries.
    \item Differential privacy provides a framework which allows to add noise to queries in a way that mitigates averaging or other attacks. Often the scale of noise required is optimum in order to avoid privacy catastrophes, e.g., reconstruction attacks.
\end{itemize}

\section{Discussion}
We reiterate that due to ethical reasons we did not demonstrate the attack on real census data, but rather showed the vulnerability of the perturbation method by applying it on a synthetic dataset. Thus, the actual TableBuilder tool is vulnerable to our attack and remains at risk from similar attacks. Notice that even though the TableBuilder tool is not equipped with an API the attack could still be performed in an automated way, e.g., one could use web-based scripts to query the tool.

We communicated the vulnerability to ABS. They acknowledged that the attack relates to TableBuilder. In response, we were told that the ABS is bringing some upcoming changes to the TableBuilder tool. These include applying user-specific cap on the number of queries (users will have to re-apply once their query quota expires), only allowing highly aggregated data to TableBuilder Guest (which can be accessed without registration), and monitoring/auditing of TableBuilder usage logs. We believe that a more controlled access to TableBuilder is definitely a step in the right direction. For instance, only allowing access to trusted users.\footnote{One may wonder if the user is trusted, why, then, use a perturbation mechanism at all? One reason provided to us is that the user may wish to publicly release information obtained from TableBuilder, e.g., a journalist. However, in such cases the information that the user wishes to publish, and only this information, can always be ``sanitised'' before publication.} As mentioned above, however, query capping/throttling and auditing are mitigation measures that are difficult to implement and impose. It is not clear exactly how many queries are safe (if noise is not a function of the number of queries), and as far as auditing is concerned, it is extremely difficult to determine if a series of queries is launched to carry out an attack or not (our attack or others). 

As we have demonstrated, our attack crucially retrieves low counts as well. Most importantly, it retrieves counts of 1 (which are suppressed by the TBE algorithm). Extracting such ``uniques'' is linked to re-identifying individuals. Some may argue that finding a unique is not the same as identifying a real person in the population exhibiting those attribute values. However, 
    maintaining this flimsy distinction between the two cases provides little solace; once uniques are identified, a little background information is enough to link them to real persons in the population~\cite{sweeney}. 

Finally, we would like to draw attention to a similar data usage scenario relating the United States (US) Census Bureau who seek to publish some aggregated form of the 2020 Census of Population and Housing~\cite{garfinkel}. The Bureau has internally investigated the applicability of database reconstruction attacks~\cite{exposed, dinur-nissim, dwork-recon, power-recon} on the 2010 (aggregated) census data and has come to the conclusion that given the amount of information leaked per person, there is a ``theoretical possibility'' that the census data could be reconstructed. We first note that these reconstruction attacks are also applicable to noisy data (where noise is significantly less than the amount of statistics released). Secondly, our averaging attack can also be seen as a form of reconstruction attack, where the attacker can reconstruct target columns of the underlying dataset. Based on their findings, Garfinkel, Abowd and Martindale~\cite{garfinkel} conclude:
    \begin{displayquote}
    Faced with the threat of database reconstruction, statistical agencies have two choices: they can either publish dramatically less information or use some kind of noise injection. Agencies can use differential privacy to determine the minimum amount of noise necessary to add, and the most efficient way to add that noise, in order to achieve their privacy protection goals.
    \end{displayquote}
    These recommendations for the US census data are inline with our suggestions for the Australian census data.






\section{Related Work}
\label{sec:rw}
Our attacks are not the only attacks reported on the TBE algorithm. A differencing attack on the TBE algorithm has been documented before, through which some information about a target individual can be inferred~\cite{chipperfield, keefe}. The attack essentially relies on some background knowledge. For instance, suppose that we know that $n$ out of $n+1$ individuals in a particular group, identified by a vector of attributes $\mathbf{b} \in \mathbf{B}$, satisfy a particular attribute value $a \in A$, where $A \notin \mathbf{B}$. Further assume that the $(n+1)$th individual, the target, has the same background, i.e., takes on the values $\mathbf{b}$, and we would like to know if the individual also exhibits $a \in A$. By querying the TBE algorithm on $\mathbf{b}$, and then $\mathbf{b} + a$, we can tell if the individual does not exhibit $a$ if the two answers returned by the TBE algorithm are different. However, notice that this attack is fundamentally different than our attacks. One major difference being that our averaging attack does not rely on any background information about other individuals.

A somewhat similar attack that exploits the use of suppression is highlighted in~\cite{culnane-opal}. Although the attack is in relation to the application of a differentially private mechanism to release histograms of transport data~\cite{opal}, it can also be applied to the TBE algorithm. For instance, suppose that the query answer $q_{a_1}$ is larger than the suppression parameter, and the query answer $q_{a_2}$ is lower, for $a_1, a_2 \in A$, for some attribute $A$. The TBE algorithm will return a noisy count for $q_{a_1}$ and 0 for $q_{a_2}$. However, the query answer $q_{a_1} \wedge q_{a_2}$ is higher than the suppression parameter. If the answer returned by TBE is different than $q_{a_1}$, then the analyst learns that $q_{a_2}$ is not zero. Indeed, this is the observation used by us in Section~\ref{subsub:broaden}.

As noted before, there is a specific class of attacks known as \emph{linear reconstruction} attacks that seeks to reconstruct a whole (target) column of a sensitive dataset~\cite{exposed, dinur-nissim, dwork-recon, power-recon} based on linear programming. Algorithms that allow overly accurate answers to too many linear queries are susceptible to these attacks. More precisely, algorithms which return noisy answers with noise bounded within $o(\sqrt{n})$, where $n$ is the number of rows (individuals) in the dataset, succumb to these attacks. However, it is not clear how linear reconstruction attacks can be applied to the TBE algorithm due to the restricted query interface. Linear reconstruction attacks use \emph{random subset sum} queries which requires the ability to query random subsets of rows of the underlying dataset~\cite{cohen-linear}. The restricted interface in TBE does not allow such queries. 

Diffix~\cite{diffix, diffix-birch} is another disclosure control mechanism that is built on some principles similar to TBE. In particular, Diffix also uses the concept of \emph{sticky noise}, e.g., same query, same answer. However, the authors in~\cite{diffix} note that a naive application of sticky noise is susceptible to a \emph{split averaging attack}. The attack is similar to our attack: ask queries on attribute values and their negations (e.g., $a$ and NOT $a$), and then average over multiple attribute values. Due to this and other attacks, Diffix uses the idea of layered sticky noise, using a combination of static noise (depending on the query conditions, e.g., gender being female) and dynamic noise depending on both the query conditions and the set of contributors~\cite{diffix, diffix-birch}. We note that the split averaging attack is different from our attack, as the TBE interface does not accept negated queries. We instead need to rely on two-partitions and the use of the total query. Furthermore, we give a detailed analysis on success probabilities as a function of the number of attribute values with non-zero counts returned by the algorithm.

Interestingly, Diffix circumvents our proposed attack by including \emph{per query condition} sticky noise. For instance, assume the attacker wants to know the exact answer to $q_b \wedge q_A$, where $b$ is an attribute value of some attribute $B$, and $A$ is the target attribute. Let $\{A_1, A_2\}$ be a two-partition of $A$. Then the noise added to $q_b \wedge q_{A_1}$ and $q_b \wedge q_{A_2}$ has a noise component added due to the contributors of the two queries, plus noise added due to the conditions $q_b$, $q_{A_1}$, and $q_{A_2}$. The noise added to the last two changes per two-partition and hence can be averaged out, but $q_b$ remains ``sticky.'' Hence the result of the averaging attack will not average this noise term out. The obvious drawback is that more noise is added per query. More details of the multi-layer noise in Diffix is given in \cite[\S 5.2]{diffix}. 

Gadotti et al.~\cite{gadotti-signal} propose a different averaging attack on Diffix, which circumvents the layered sticky noise by first removing the static noise component and then uses averaging to distinguish between two probability distributions, one with the sensitive attribute set to 1 and the other where it is set to 0. The attack relies on knowing whether the target record is unique in the dataset, which is likely to be true for a significant fraction of records in the dataset~\cite{gadotti-signal}. A more involved attack, called cloning attack, uses ``dummy'' conditions in queries to obtain the same set of contributors~\cite{gadotti-signal}. However this attack relies on the richness of the query language. 

There are also reconstruction attacks reported on Diffix~\cite{cohen-linear}, which rely on the ability to select random ``enough'' rows from the underlying dataset by exploiting the rich query interface of Diffix. In the most recent version of Diffix, the attack has been seemingly mitigated by restricting queries that would isolate individuals in the dataset~\cite{gadotti-signal, aircloak}. As mentioned before, the TableBuilder interface is highly restrictive, and it is not clear how such reconstruction attacks could be carried out on the TBE algorithm.

\section{Conclusion}
We have shown an averaging attack that retrieves actual values exhibited by an attribute (or one or more of its attribute values) in a dataset which can only be accessed via a privacy-preserving algorithm that adds bounded uniform noise to the answers. In line with previous research on linear reconstruction attacks (see e.g.,~\cite{exposed}), we show that if the number of allowed queries are above a given mark, the algorithm fails to provide privacy. We have demonstrated the attack on a synthetic dataset accessed via the TBE algorithm used for the ABS TableBuilder. While the TBE algorithm might be patched to resist the particular attack mentioned in this paper, we would like to stress that this may not be the only attack possible. A better alternative is to scale noise according to the number of queries allowed to minimise information leakage from a theoretical point of view~\cite{dinur-nissim}. This will guarantee that privacy is maintained in practice. We also restate that we have only considered one subset of queries (counting queries), and the attack may be applicable to other types of queries, e.g., range queries on continuous data. 

\bibliographystyle{plain}
\bibliography{nonoise-ref}

\appendix




\section{Proofs}
\label{app:proofs}

\subsubsection*{Proof of Proposition~\ref{prop:bnc}}
\begin{proof}
Let $e$ be the noise added by $\mathcal{M}$. (a) First assume $n > s$. Then since $e \in [-r , r]$, $\alpha = n + e \ge n - r \ge n - s > 0$. Here we have used the fact that $s \ge r$. If $n \le s$, then $\alpha = 0$ (step 2 of the algorithm). (b) Again, first assume $n > s$. Then $n - r \le \alpha \le n + r \Rightarrow n - s \le \alpha \le n + r$. Now, if $n \le s$, then $\alpha = 0$. Trivially, $n - s \le 0 = \alpha$. (c) If $C(q)$ is empty, then $|C(q)| = 0 < s$, and hence $\mathcal{M}$ should always return a 0 in this case.
\end{proof}

\subsubsection*{Proof of Lemma~\ref{lem:e:ind}}
\begin{proof}
Since $n_1$ and $n_2$ are both greater than $r$, the noisy answers returned by Bounded Noisy Counts are non-zero. Furthermore, since $n_1, n_2 > r > 0$, we see that $\{C(q_b \wedge q_{a_1}), C(q_b \wedge q_{a_2})\}$ is a partition of $C(q_b)$. Hence, the two have necessarily different contributors: $C(q_b \wedge q_{a_1}) \neq C(q_b \wedge q_{a_2})$. Therefore, Bounded Noisy Counts adds independent noise to the corresponding queries. Furthermore, $C(q_b) \ne C(q_b \wedge q_{a_1})$ and $C(q_b) \ne C(q_b \wedge q_{a_2})$, since the cardinality of both are greater than $r$, and hence cannot be equal to the total. Therefore, there is independent noise added to $n$ as well.
\end{proof}

\subsubsection*{Proof of Lemma~\ref{lem:diff:cont}}
\begin{proof}
Assume the contrapositive for some $i$ and $j$. 
Then since $C(q_{b_i}) = C(q_{b_j})$, the Bounded Noisy Counts algorithm should add the same noise to $q_{b_i}$ and $q_{b_j}$, and hence $\mathcal{M}(q_{b_i}) = \mathcal{M}(q_{b_j})$; a contradiction.
\end{proof}

\subsubsection*{Proof of Lemma~\ref{lem:twenty}}
\begin{proof}
First let us consider the number of permutations whose sum is greater than $3(r - 1)$. Note that none of the $E_i$'s can be less than $r - 3$. To see this, note that if $r = 1$, then any of the $E_i$'s cannot be equal to $r - 3 = -2$. Let us assume that $r > 1$, then if any of the $E_i$'s is $\le r - 3$, then the maximum possible sum is $\le r - 3 + r + r = 3(r-1)$, which is our threshold. Thus, we enumerate all possible permutation of values of the $E_i$'s, such that $E_i \ge r - 2$ and $E_1 + E_2 + E_3 >  3(r - 1)$. This is shown below.

\subsubsection*{Proof of Proposition~\ref{prop:pertsucc}}
\begin{proof}
From Lemma~\ref{lem:twenty}, there are exactly 20 possible values of the tuple $(z_1, z_2, z_3)$, for which $z$ in step 4 of the algorithm has sum greater than $3(r-1)$ or less than $-3(r-1)$.\footnote{Note that even though the lemma applies to the sum $z_1 + z_2 + z_3$, it is easy to see that it also holds true for $z_1 + z_2 - z_3$.} Through Lemma~\ref{lem:e:ind} the variables $z_i$ are i.i.d. The probability that $z$ for the $i$th set of queries to the Attribute Analyzer is within the interval $[-3(r-1), 3(r-1)]$ is given by $1 - 20/(2r+1)^3$. The result follows for all $m$ attributes, since the $i$th $z$ in step 4 is independently distributed due to Lemma~\ref{lem:diff:cont}.   
\end{proof}

\begin{center}
\begin{tabular}{c|c|c} 
$E_1$ & $E_2$ & $E_3$ \\
 \hline
 $r$ & $r$ & $r$ \\ 
 $r$ & $r$ & $r-1$ \\
 $r$ & $r$ & $r-2$ \\ 
 $r$ & $r-1$ & $r$ \\ 
 $r$ & $r-1$ & $r-1$ \\ 
 $r$ & $r-2$ & $r$ \\
 $r-1$ & $r$ & $r$ \\
 $r-1$ & $r$ & $r-1$ \\
 $r-1$ & $r-1$ & $r$ \\
 $r-2$ & $r$ & $r$ \\
\end{tabular}
\end{center}

There are exactly 10 such values. By symmetry, the same holds for $E_1 + E_2 + E_3 < -3(r - 1)$.
\end{proof}

\subsubsection*{Proof of Lemma~\ref{lem:numtwoparts}}

\begin{proof}
There are $2^m$ elements in the power set of $A$. Out of these, two are $\emptyset$ and $A$ itself. Out of the remaining $2^m - 2$ elements (subsets of $A$), we can construct a two-partition by choosing any element as the first subset, say $A'$, and $A - A'$ as the other subset. Since $A - A'$ is also in the power set, we have counted each two-partition twice. Thus, dividing $2^m - 2$ by 2 gives us the result.
\end{proof}

\subsubsection*{Proof of Lemma~\ref{lem:twopart}}

\begin{proof}
Let $S = A' \Delta A''$.\footnote{For any two sets $A$ and $B$, $A \Delta B$ denotes their symmetric difference.} Then, since $A' \ne A''$, $S \ne \emptyset$. Also, since $\mathcal{M}(q_\mathbf{b} \wedge q_{a_i}) \ne 0$ for all $i \in [m]$, it follows from Proposition~\ref{prop:bnc}, that $C(q_\mathbf{b} \wedge q_{A'})$ and $C(q_\mathbf{b} \wedge q_{A''})$ are non-empty sets. Now, assume to the contrary that $C(q_\mathbf{b} \wedge q_{A'}) = C(q_\mathbf{b} \wedge q_{A''})$. Since any contributor (user) can have only one attribute value in $A$, it must follow that for this contributor to be in both sets of contributors, the attribute value should be in the intersection of both $A'$ and $A''$. Since the two sets of contributors are equal, this means that any attribute value $a \notin A' \cap A'' = S$ is not taken by any contributor. Therefore, $(q_\mathbf{b} \wedge q_{a})(D) = 0$, for all $a \in S$, which means that Bounded Noisy Counts outputs $\mathcal{M}(q_\mathbf{b} \wedge q_{a}) = 0$. But this contradicts the assumption. 
\end{proof}

\subsubsection*{Proof of Lemma~\ref{lem:twopart:ind}}

\begin{proof}
Fix an $i$. The the two error variables $E_1^{(i)}$ and $E_2^{(i)}$ are noise terms added to the two sets of the corresponding two-partitions. Since by assumption all counts from Bounded Noisy Counts are non-zero, the queries corresponding to the two sets have non-zero set of contributors (Proposition~\ref{prop:bnc}). Since the two queries are on disjoint partitions, they also have different contributors (in fact, mutually exclusive). Therefore, the Bounded Noisy Counts algorithm adds independent noise with distribution $\mathbb{U}_{\pm r}$. Now consider, $Z_1, \ldots, Z_k$. By Lemma~\ref{lem:twopart} the query corresponding to every set in  the set of two-partitions $P_A$ has different contributors. Hence $E_j^{(i)}$ are i.i.d. with distribution $\mathbb{U}_{\pm r}$, where $j \in \{0, 1\}$, $i \in [k]$. The result follows.
\end{proof}

\subsection{Lower Bound on the Success Probability of Noise Remover}
\label{appsub:lower-bound}
It is easy to see that $\mathbb{E} (E_j^{(i)}) = 0$. And by the variance of the discrete uniform distribution
\[
\text{Var} (E_j^{(i)}) = \frac{(r - (-r) + 1)^2 - 1}{12} = \frac{r(r+1)}{3}.
\]
By the linearity of expectation
\[
\mathbb{E}(Y_i) = 0,
\]
and by Lemma~\ref{lem:twopart:ind},
\[
\text{Var}(Y_i) = \frac{2r(r+1)}{3}.
\]
Again through linearity of expectation
\[
\mathbb{E}(\overline{Y}) = 0,
\]
and by Lemma~\ref{lem:twopart:ind},
\[
\text{Var}(\overline{Y}) = \frac{1}{k^2} \sum_{i = 1}^k \frac{2r(r+1)}{3} = \frac{2}{3}\frac{r(r+1)}{k}.
\]
Using Chebyshev's inequality, we see that
\begin{equation*}
\Pr \left( \left\lvert \overline{Y} - \mathbb{E}(\overline{Y}) \right\rvert \ge \epsilon \right) \le \frac{\text{Var}(\overline{Y})}{\epsilon^2}.
\end{equation*}
By setting $\epsilon = 0.5$, and putting in the values of $\mathbb{E}(\overline{Y})$ and $\text{Var}(\overline{Y})$, we get
\[
\Pr \left( \left\lvert \overline{Y} \right\rvert \ge 0.5 \right) \le \frac{8}{3}\frac{r(r+1)}{k}.
\]
Thus,
\begin{equation*}
\Pr \left( \left\lvert \overline{Y}  \right\rvert < 0.5 \right) \ge 1 - \frac{8r(r+1)}{3(2^{m-1} - 1)}.    
\end{equation*}

\subsection{Exact Success Probability of Noise Remover}
\label{appsub:exact}
Let $X_1 = E_1$, and define $X_i = X_{i - 1} + E_i$, for $i \in \{2, \ldots, 2k\}$. Then the probability mass function of $X_2$ is given by
\[
f_{X_2}(x) = \sum_{y = -\infty}^{+\infty} f_{X_1}(y) f_E(x - y) = \sum_{y = -\infty}^{+\infty} f_{E}(y) f_E(x - y),
\]
and for every $i$
\[
f_{X_i}(x) = \sum_{y = -\infty}^{+\infty} f_{X_{i-1}}(y) f_E(y - x). 
\]
Thus, we can iteratively find $f_{X_{2k}}$, the probability mass function of $X_{2k}$. Now,
\begin{align}
\Pr \left( \left\lvert \overline{Y}  \right\rvert < 0.5 \right) &= \Pr \left( \left\lvert \frac{X_{2k}}{k}  \right\rvert < 0.5 \right) \nonumber \\
																					&= \Pr \left( -\frac{1}{2} < \frac{X_{2k}}{k}  < \frac{1}{2} \right) \nonumber \\
																					&= \Pr \left( -\frac{k}{2} < {X_{2k}}< \frac{k}{2} \right) \nonumber \\
																					&= \sum_{x \in (-k/2, k/2)} f_{X_{2k}} (x) \nonumber.
\end{align}

\section{Synthetic Data Details}
\label{app:syndata}

\begin{table*}[!ht]
\centering
\caption{The distribution of counts in the target column in the synthetic dataset against each attribute value $a \in \{1, 2, \ldots, 107\}$. ``Counts'' represents true counts, and ``TBE'' represents noisy counts from the TBE algorithm with parameter $r = 2$.}
\label{tab:syndata}
\resizebox{\textwidth}{!}{%
\begin{tabular}{c|c|c|c|c|c|c|c|c|c|c|c|c|c|c|c|c|c|c|c|c}
$a$ & 1 & 2 & 3 & 4 & 5 & 6 & 7 & 8 & 9 & 10 & 11 & 12 & 13 & 14 & 15 & 16 & 17 & 18 & 19 & 20 \\
\hline
Counts & 1 & 3 & 6 & 12 & 33 & 53 & 114 & 199 & 372 & 677 & 1075 & 1837 & 2884 & 4388 & 6496 & 9136 & 12694 & 16893 & 21513 & 26566 \\
TBE & 0 & 0 & 6 & 12 & 34 & 52 & 116 & 197 & 373 & 678 & 1074 & 1838 & 2883 & 4389 & 6495 & 9137 & 12692 & 16892 & 21515 & 26564 \\ 
\hline
$a$ & 21 & 22 & 23 & 24 & 25 & 26 & 27 & 28 & 29 & 30 & 31 & 32 & 33 & 34 & 35 & 36 & 37 & 38 & 39 & 40 \\
\hline
Counts & 31854 & 36741 & 40268 & 43426 & 44865 & 44812 & 43054 & 40259 & 35698 & 31534 & 26103 & 20953 & 16539 & 12430 & 8977 & 6297 & 4283 & 2715 & 1775 & 1085 \\ 
TBE & 31853 & 36739 & 40267 & 43427 & 44867 & 44813 & 43053 & 40258 & 35696 & 31532 & 26105 & 20955 & 16537 & 12431 & 8975 & 6296 & 4284 & 2713 & 1774 & 1084 \\
\hline
$a$ & 41 & 42 & 43 & 44 & 45 & 46 & 47 & 48 & 49 & 50 & 51 & 52-107 &&&&&&&&\\
\hline
Counts & 614 & 377 & 196 & 93 & 53 & 24 & 14 & 3 & 4 & 1 & 1 & 0 &&&&&&&&\\
TBE & 615 & 378 & 195 & 95 & 52 & 25 & 15 & 0 & 0 & 0 & 0 & 0 &&&&&&&&\\
\end{tabular}%
}
\end{table*}

The synthetic dataset used in our experimental evaluation in Section~\ref{sec:eval} was generated by first fixing $n = 600,000$ rows. Next we generated $n$ normally distributed samples with mean 25 and standard deviation 5. The resulting samples were then binned into their nearest integer bins labelled $1$ to $51$. We assumed an attribute $A$ with $107$ different attribute values. We assigned the counts in the bins $1$ to $51$ to the attribute values $a = 1$ to $51$, respectively. Attribute values $52$ to $107$ were fixed at 0. For our attack evaluation, we only used this one column, where the counts are normally distributed. Table~\ref{tab:syndata} shows the counts in the target column used in our experimental evaluation together with the noisy counts retrieved via the TBE algorithm with perturbation parameter $r = 2$. 

\section{Constraints on the Dataset}
\label{app:constraints} The Noise Remover algorithm requires a minimum number of attribute values with non-zero counts returned by the Bounded Noisy Counts for a given probability of success. Abusing terminology, we call them \emph{attribute values with non-zero outputs}. For an attribute $A$ with $m$ attribute values with non-zero outputs, the probability of successfully determining the query answer $q_A$ can be determined via Eq.~\ref{eq:nr:suc:exact}. Recall that $m$ attribute values with non-zero outputs enables $k = 2^{m - 1} - 1$ two-partitions (cf. Section~\ref{subsub:broaden}). The higher the number of two-partitions available, the closer we get to the true answer as a function of the number of queries $t \le 2k$ to Bounded Noisy Counts. We are interested in evaluating the success probability when $m$ is small. To do this, we ran an experiment where we vary the number of attribute values of an attribute $A$ with non-zero outputs from 2 to 11. Then for each $m$, we run the Noise Remover with all the $k = 2^{m-1} - 1$ two-partitions, and obtain the result as the guess for $q_A$. We ran the experiment 10,000 times for each perturbation parameter $r$ in the range $2$ to $10$. The results are shown in Figure~\ref{fig:att-needed}. Note that $m = 2$ means that we have only 1 two-partition. In this case, the noise remover's success rate is simply $1/(2r + 1)$, the probability of guessing the true answer. For small parameters, $r \le 5$, we need 9 attribute values with non-zero outputs for the attack to be successful with more than 90 percent success rate. This reaches to $11$ for the perturbation parameter $r = 10$.

\begin{figure}[ht!]
\centering
\ifpets
\else
\includegraphics[scale=0.8]{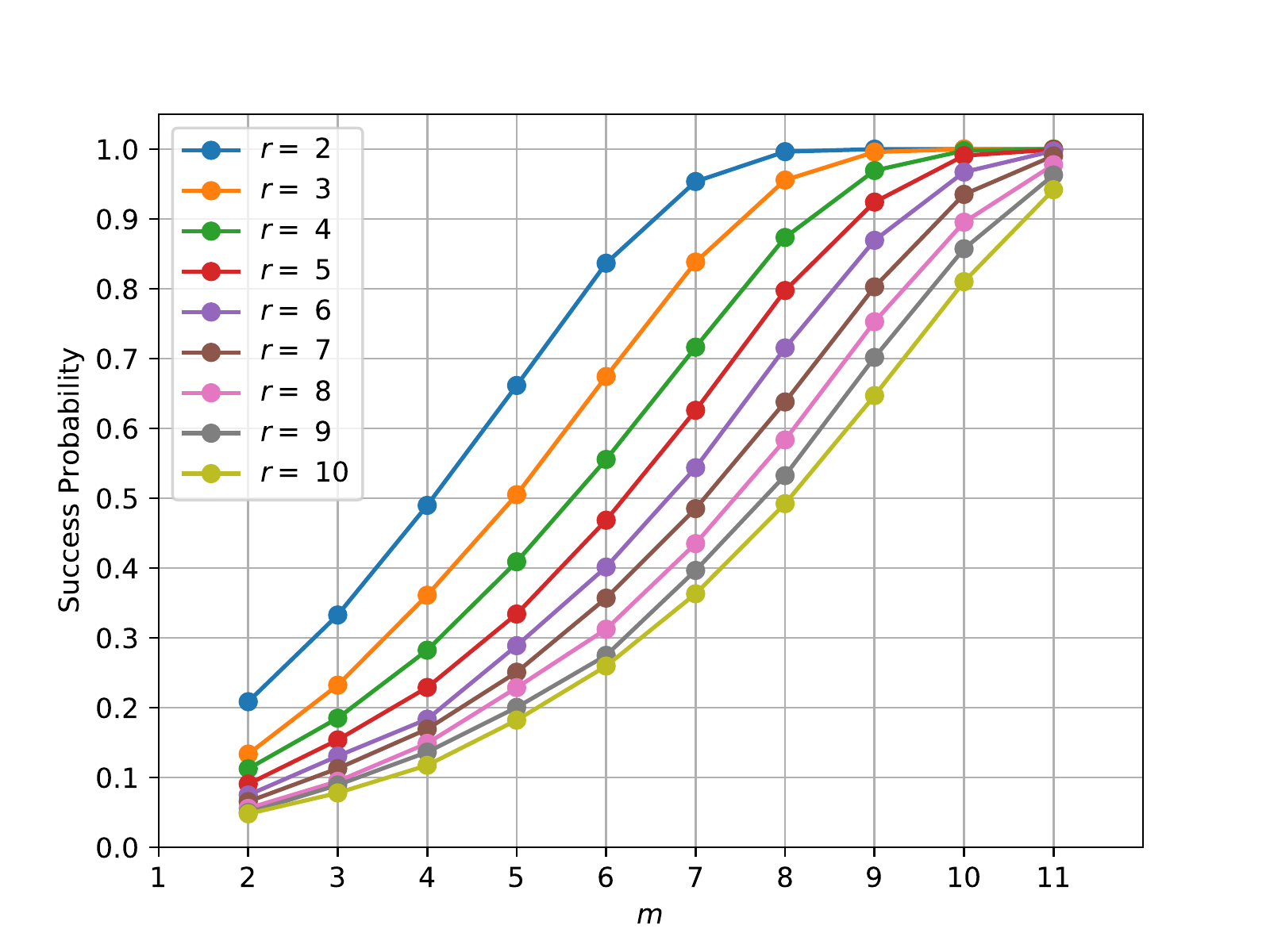}
\fi
\caption{Probability of successfully retrieving the count of the query $q_{A}(D)$ of an attribute $A$ with $m$ different attribute values $a_1, \ldots, a_m$ such that the answers returned by Bounded Noisy Counts is non-zero for all $q_{a_i}$.}
\label{fig:att-needed}
\end{figure}

However, there are several workarounds when we have less than the ideal number of attribute values with non-zero outputs. We discuss two of them. First, we note that the attack can still be launched if we have only two attribute values with non-zero outputs. Let us assume the attribute $A = \{a_1, a_2, e_1, \ldots, e_m\}$. Assume that the queries on attribute values $a_1$ and $a_2$ have non-zero outputs. Further assume that the answers to the queries $q_{e_i}$, for $i \in [m]$, is zero from Bounded Noisy Counts, but they have \empty{non-empty} contributors. We denote this subset by $A'$. Then, first we can compute $2^{m-1} - 1$ two-partitions of $A'$. In each partition, we add $a_1$ to the first set and $a_2$ to the second set. Then it is easy to see that the resulting $2^{m-1} - 1$ partitions can be used to launch the noisy removing attack, as the corresponding query answers will not be suppressed and will have fresh noise added. Of course, for this to work we need to confirm if the true answers to the $q_{e_i}$'s are indeed non-zero. This can be done by using $a = a_1$ or $a = a_2$ and obtaining answers to $q_{a}$ and $q_a \wedge q_{e_i}$. If the two noisy answers are different, then necessarily $C(q_{e_i})$ is non-empty. The probability that the answer would be different is given by $1 - (2r+1)^{-2}$. 

Second, the attack can still be launched by adding/removing attribute values from other attributes. For instance, if we are interested in knowing $q_a$ for some attribute $a \in A$, and we have another attribute $B$, with $m$ attribute values with non-zero outputs, we can create two-partitions of $B$. Then for each two-partition $\{B', B''\}$, we can sum the noisy query answers $q_a \wedge q_{B'}$ and $q_a \wedge q_{B''}$. 

Thus, while a minimum number of attribute values with non-zero outputs in the target attribute makes the attack simpler, there are other ways to carry out the attack. This means that the attack cannot be rendered ineffective simply because the target attribute has a low number of attribute values with non-zero outputs. 

\end{document}